\documentclass[11pt,reqno]{amsart}
\usepackage{txfonts}
\usepackage{mathrsfs}
\usepackage{amsfonts}
\allowdisplaybreaks[4]
\usepackage{graphicx} %We can use any other package if it is necessary
\usepackage{epsfig}
\usepackage{bbm}
\usepackage{amssymb}
\usepackage{amsmath,xypic}
\usepackage{cite,color}
\newtheorem{theorem}{Theorem}

\newtheorem{proposition}[theorem]{Proposition}
\newtheorem{corollary}[theorem]{Corollary}

\newtheorem{remark}[theorem]{Remark}
\newtheorem{lemma}[theorem]{Lemma}
\newcommand{\be}{\begin{equation}}
\newcommand{\ee}{\end{equation}}
\newcommand{\bea}{\begin{eqnarray}}
\newcommand{\eea}{\end{eqnarray}}
\newcommand{\ba}{\begin{array}}
	\newcommand{\ea}{\end{array}}
\newcommand{\bean}{\begin{eqnarray*}}
	\newcommand{\eean}{\end{eqnarray*}}

\begin{document}

\title{Equivalence of two constructions for $\widehat{sl}_2$--integrable hierarchies}
\author{Panpan Dang, Yajuan Li, Yuanyuan Zhang, Jipeng Cheng$^*$}
\dedicatory { School of Mathematics, China University of
Mining and Technology, \\Xuzhou, Jiangsu 221116, P.\ R.\ China}
\thanks{*Corresponding author. Email: chengjp@cumt.edu.cn, chengjipeng1983@163.com.}
\begin{abstract}
In this paper, we investigate the equivalence of Date--Jimbo--Kashiwara--Miwa (DJKM) construction and Kac--Wakimoto (KW) construction for $\widehat{sl}_2$--integrable hierarchies. DJKM method has gained great success in constructions of integrable hierarchies corresponding to classical ABCD affine Lie algebras, while the KW method is more applicable, which can be even used in exceptional EFG affine Lie algebras. But in KW construction, it is quite difficult to obtain Lax equations for the corresponding integrable hierarchies, while in DJKM construction, one can derive Lax structures for many integrable hierarchies. It is still an open problem for the derivation of Lax equations from bilinear equations in KW construction. Therefore if we can show the equivalent DJKM construction for the integrable hierarchies derived by the KW construction, then it will be helpful to get corresponding Lax structures. Here the equivalence of DJKM and KW methods is showed in the $\widehat{sl}_2$--integrable hierarchy for principal and homogeneous representations by using the language of the lattice vertex algebras.\\
\textbf{Keywords}: $\widehat{sl}_2$--integrable hierarchy; Kac--Wakimoto construction; Date--Jimbo--Kashiwara--Miwa construction; bilinear equations; Lax equations; lattice vertex algebra; 

\end{abstract}

\maketitle
\tableofcontents

\section{Introduction}
The infinite dimensional Lie algebras play an important role in the study of integrable systems \cite{Alexandrov2021,DATE,DRINFELD,IKEDA,JIMBO,KAC, KAC1989,KAC2003,KAC2013,KACs ,Carpentier2020,Kac-Van,Miwa2000,Liu2022,Liu2020,Liu2015, Wu2017,Kac2023,Bakalov2015,Billig1999,Casati2006,CHENG2021,Wu2010, Kac1990}. Many famous results have been done in this field by different schools, including Date--Jimbo--Kashiwara--Miwa (DJKM) in Kyoto school \cite{DATE,Miwa2000,JIMBO},  Kac--Wakimoto (KW) \cite{KAC1989} in Kac school and Drinfeld--Sokolov (DS) \cite{DRINFELD} in Russia school. In Kyoto school, they mainly use the free fermions and boson--fermion correspondence \cite{DATE,JIMBO,KAC,KAC2003,KAC2013, Alexandrov2013,Kac-Van,Miwa2000,Kac1990,Kroode} to construct the integrable hierarchies for the classical affine Lie algebras, which is called the DJKM method. The major element in DJKM method is the operator $S^{\pm}$ in terms of the charged free fermions $\psi_i^{\pm}\left(i\in\mathbb{Z}+1/2\right)$, defined by
\[
\begin{split}
S^{\pm}=\sum_{i\in\mathbb Z+1/2}\psi_i^{\pm}\otimes\psi_{-i}^{\mp},
\end{split}
\]
where $\psi_i^{\pm}\psi_j^{\pm}+\psi_j^{\pm}\psi_i^{\pm}=0,\ \psi_i^{\pm}\psi_j^{\mp}+\psi_j^{\mp}\psi_i^{\pm}=\delta_{i+j,0}.$
Since there is no fermionic realization for the exceptional EFG affine Lie algebras \cite{Alexandrov2013, KAC1989}, DJKM method can not work in this case. Also there is no unified way to construct $S^{\pm}$. Thus in Kac--Wakimoto's famous work  \cite{KAC1989}, the DJKM method is improved by using the Casimir operator $\Omega$ of the affine Lie algebras instead of the operator $S^{\pm}$ in DJKM construction. By KW method \cite{KAC1989}, one can theoretically construct the integrable hierarchies corresponding to any affine Lie algebras. The integrable hierarchies constructed by DJKM and KW methods \cite{DATE,JIMBO,KAC1989,KAC,KAC2003, Kac-Van,Bakalov2015,Billig1999, Casati2006,Milanov2016,Wu2010,Kac1990} are usually presented in the form of Hirota bilinear equations in terms of tau function, which rely on the vertex operator representations of the affine Lie algebras. While in Drinfeld--Sokolov construction \cite{DRINFELD,Liu2022,Liu2020,Liu2015,Hollowood1993,Wu2017, Liu2011}, the principal gradations of the affine Lie algebras are very essential and the corresponding integrable hierarchies are usually in the form of Lax equations. It is very interesting to make clear the relations among DJKM, KW and DS constructions. For DS and KW, it is proved in \cite{Hollowood1993} that there is one to one correspondence for part of integrable hierarchies constructed in these two methods. In \cite{CHENG2021,Liu2011}, the principal Kac--Wakimoto hierarchy of D--type is showed to be the Drinfeld--Sokolov hierarchy of D--type, that is one kind of reduction of $2$--component BKP hierarchy. Here we would like to investigate the relations of DJKM and KW constructions for the $\widehat{sl}_2$--integrable hierarchies.

The derivation of Lax equations from bilinear equations is usually very important, since many important integrable properties rely on the Lax structure \cite{Faddeev2007, Babelon2003}, such as Hamiltonian structures and \cite{Hollowood1993,KAC1989} constructed quantities. By now there is still no unified way to derive the Lax equations from bilinear equations. In KW construction \cite{CHENG2021,KAC1989,Wu2010,Milanov2016}, the major difficulty comes from the term
$$c\otimes d+d\otimes c$$
in Casimir operator $\Omega$ of the affine Lie algebras (e.g. Subsection \ref{jsyy1}), which is usually corresponding to a special term in Hirota bilinear equations like
$$\sum_{j\geq0}(2j+1)\left(t'_{2j+1}-t''_{2j+1}\right) \left(\partial_{t'_{2j+1}}
-\partial_{{t''_{2j+1}}}\right)
\left(\tau(t')\tau(t'')\right)$$
in $\widehat{sl}_2$--integrable hierarchy (See Proposition \ref{js} in Subsection \ref{js-2}). To our best knowledge, there is still no successful case for the derivation of the Lax equation directly from the bilinear equation in KW construction. But in DJKM construction, there are many examples, e.g., KP and Toda hierarchies \cite{DATE,Takasaki2018}, whose Lax equations can be obtained from bilinear equations. Therefore if we can find the equivalent DJKM construction for the integrable hierarchy derived by the KW method, it is possible for us to derive the corresponding Lax equations.

For the $\widehat{sl}_2$--integrable hierarchies, it is the KdV hierarchy in principal representations, while in homogenous case, it is the $1$--Toda hierarchy. In \cite{Bakalov2015}, the integrable hierarchies corresponding to the representations for $\widehat{sl}_2$ of any level $k$ are constructed by KW method. In \cite{Billig1999}, the $\widehat{sl}_2$--integrable hierarchies are generalized into the case of toroidal Lie algebra by KW method, while the corresponding Lax equations are constructed in \cite{IKEDA} as the Bogoyavlensky--KdV hierarchy. In \cite{Casati2006}, the polynomial Lie algebra generalization of $\widehat{sl}_2$ is also used to construct the integrable hierarchy by KW method, which is the coupled KdV hierarchy. Here in this paper, we will show the equivalence of DJKM and KW methods in $\widehat{sl}_2$--integrable hierarchy by using the language of lattice vertex algebras for the principal and homogenous cases.

The rest of this paper is organized in the way below. In Section $2$, the lattice vertex algebra is firstly reviewed, then $\widehat{sl}_2$ is realized as a subalgebra of lattice vertex algebra. In Section $3$, the untwisted $\widehat{sl}_2$--integrable hierarchy is constructed by the
standard KW construction. Then by using the decomposition of Casimir operator, the untwisted $\widehat{sl}_2$--integrable hierarchy is showed to be $1$--Toda hierarchy. In Section $4$, we firstly review twisted module over lattice vertex algebra, then the twisted representation of $\widehat{sl}_2$ is constructed. Based upon this, we discuss the equivalence of KW and DJKM constructions in twisted $\widehat{sl}_2$--integrable hierarchy. Finally, some conclusions and discussions are given in Section $5$.

\section{Realization of $\widehat{sl}_2$ on lattice vertex algebras}
In this section, we will firstly review some basic facts about lattice vertex algebra. One can refer to \cite{KACs} for more details. Then the Lie algebra $\widehat{sl}_2$ is realized as the subalgebra of the lattice vertex algebra. Based upon this, the Casimir operator of $\widehat{sl}_2$ is expressed by using lattice vertex algebra.

Let $Q=\mathbb{Z}v_1\oplus \mathbb{Z}v_2\oplus\dots\oplus \mathbb{Z}v_{\mathcal{N}}$ be an integral lattice of rank $\mathcal{N}$ with the following symmetric bilinear form
\begin{align*}
\left(v_i|v_j\right)=\delta_{ij},
\end{align*}
where $i,j=1,2,\cdots,\mathcal{N}$. Then define $\mathbb{C}_{\varepsilon}(Q)$ as a group algebra with basis $e^{\alpha}\ (\alpha\in Q)$ and the multiplication $$e^{\alpha}e^{\beta} =\varepsilon(\alpha,\beta)e^{\alpha+\beta},$$
where $\varepsilon(\cdot,\cdot)$ is a bi-multiplication function on $Q\times Q$ given by
\begin{align}\label{varepsilon}
\varepsilon\left(v_i,v_j\right)=\left\{
    \begin{array}{ll}
      1, &  i< j,\\
      -1, &  i\geq j.
    \end{array}
  \right.
\end{align}
Set $S$ to be the symmetric tensor algebra of $S\left(\mathfrak{h}\left[\lambda^{-1}\right]\lambda^{-1}\right)$ with $\mathfrak{h}$ be the complexification of $Q$. Then the lattice vertex algebra is defined by
$$V_Q=S\otimes \mathbb{C}_{\varepsilon}(Q).$$

The state--field correspondence of $V_Q$ is given by
$$Y\left(\alpha^1_{(-m_1)}\cdots\alpha^n_{(-m_n)}e^{\beta},z\right)=:\partial_z^{(m_1)}\alpha^1(z)\cdots\partial_z^{(m_n)}\alpha^n(z)Y\left(e^{\beta},z\right):,$$ where $m_j\geq0,\ \partial_z^{(m)}=\partial_z^m/m!$,\ $\alpha^i,\ \beta\in Q$ and $\alpha^i_{(n)}=\alpha^i\lambda^n$. Here the normally ordered product is defined inductively from right to left and
\begin{align}\label{zhyy}
:a(z)b(z):=a(z)_+b(z)+p(a,b)b(z)a(z)_-,
\end{align}
where $a(z)_+$ means the part with nonnegative $z$ powers, while $a(z)_-$ is the part with negative $z$ powers. And $p(a,b)=(-1)^{p(a)p(b)}$, where $p(a)$ and $p(b)$ are the parities of $a$ and $b$ respectively. The fields $h(z)\ (h\in \mathfrak{h})$ and $Y\left(e^{\alpha},z\right)\ (\alpha\in Q)$ are given by
\begin{align}\label{field}
h(z)=\sum_{n\in\mathbb{Z}}h_{(n)}z^{-n-1},\quad Y\left(e^{\alpha},z\right)=e^{\alpha}z^{\alpha_{(0)}}{\rm exp}\left(\sum_{n<0}\alpha_{(n)}\frac{z^{-n}}{-n}\right){\rm exp}\left(\sum_{n>0}\alpha_{(n)}\frac{z^{-n}}{-n}\right).
\end{align}
Parities of $h(z)$ and $e^{\alpha}$ are defined by $p\left(h(z)\right)=0,\ p\left(e^{\alpha}\right)=|\alpha|^2\ {\rm mod}\ 2$. The vacuum of $V_Q$ is $|0\rangle=1\otimes1$ and the infinitesimal translation operator $T$ is given by
$$[T,h_{(m)}]=-mh_{(m-1)},\quad Te^{\alpha}=\alpha_{(-1)}e^{\alpha}.$$
The lattice vertex algebra $V_Q$ has a conformal vector $$\nu=\frac{1}{2}\sum_{i=1}^{\mathcal{N}}v_{i(-1)}v_{i(-1)}|0\rangle.$$

The action of $h_{(m)}$ and $e^{\alpha}$ on $V_Q$ are given by
\begin{align}\label{healphazy}
h_{(m)}\left(s\otimes e^{\beta}\right)=\left(h_{(m)}+\delta_{m,0}(h|\beta)\right)s\otimes e^{\beta},\quad
e^{\alpha}\left(s\otimes e^{\beta}\right)=\varepsilon(\alpha,\beta)s\otimes e^{\alpha+\beta},
\end{align}
where the action of $h_{(m)}$ on $S$ is determined by $h_{(m)}1=0\ (m\geq0)$ and
\begin{align}\label{hh'}
\left[h'_{(m)},h''_{(n)}\right]=m\delta_{m+n,0}(h'|h''),\quad h',h''\in\mathfrak{h}.
\end{align}
Thus for $j\geq0$
\begin{align}\label{hh'nj}
h'_{(j)}h''=h'_{(j)}h''_{(-1)}|0\rangle=\delta_{j,1}(h'|h'')|0\rangle.
\end{align}
From \eqref{healphazy}, we can find
\begin{align}\label{healpha} \left[h_{(m)},e^{\alpha}\right]=\delta_{m,0}(h|\alpha)e^{\alpha}.
\end{align}

Recall the commutator of $Y(a,z)=\sum_{n}a_{(n)}z^{-n-1}$ and $Y(b,w)=\sum_{n}b_{(n)}w^{-n-1}$ can be computed by
\begin{align}\label{YYjhgx}
[Y(a,z),Y(b,w)]=\sum_{j=0}^{+\infty}Y\left(a_{(j)}b,w\right)\partial_w^{(j)}\delta(z-w),
\end{align}
which is equivalent to
\begin{align}\label{abjhgx}
\left[a_{(m)},b_{(n)}\right]=\sum_{j=0}^{+\infty}\binom{m}{j}\left(a_{(j)}b\right)_{(m+n-j)}.
\end{align}
Here $\delta(z,w)=z^{-1}\sum_{n\in\mathbb{Z}}\left(w/z\right)^n$. By \eqref{field} \eqref{healphazy} and \eqref{abjhgx}, one can find for $j\geq0$
\begin{align}\label{heee}
h_{(j)}e^{\beta}=\delta_{j,0}(h|\beta)e^{\beta},\quad
e^{\alpha}_{(j)}e^{\beta}={\rm Res}_{z=0}z^jY(e^{\alpha},z)e^{\beta}
=\varepsilon\left(\alpha,\beta\right)p_{-(\alpha|\beta)-j-1}\left(\widetilde{\alpha}\right)e^{\alpha+\beta},
\end{align}
where $h\in\mathfrak{h},\ \alpha,\beta\in Q,\ \widetilde{\alpha} =\left(\alpha_{(-1)},\alpha_{(-2)}/2,\cdots\right)$ and $p_n(t)$  is the Schur polynomial determined by ${\rm exp}(\xi(t,z))=\sum_{n=0}^{\infty}p_n(t)z^n$ with $\xi(t,z)=\sum_{i=1}^{+\infty}t_iz^i$. Thus according to \eqref{abjhgx},
\begin{align}\label{alphaebeta}
\left[h_{(m)},e^{\beta}_{(n)}\right]=\left(h\big|\beta\right)e^{\beta}_{(m+n)},\quad \left[e^{\alpha}_{(m)},e^{\beta}_{(n)}\right] =\sum_{j=0}^{+\infty}\varepsilon(\alpha,\beta)\binom{m}{j} \left(p_{-(\alpha|\beta)-j-1}\left(\widetilde{\alpha}\right)e^{\alpha+\beta}\right)_{(m+n-j)}.
\end{align}
In particular by \eqref{YYjhgx} and the first relation in \eqref{alphaebeta}, we can obtain
\begin{align}\label{ebetah}
e^{\beta}_{(j)}h=-(h|\beta)\delta_{j,0}e^{\beta}.
\end{align}

Next let us compute $\left[h_{(m)},\nu_{(1)}\right]$ and $\left[e^\alpha_{(m)},\nu_{(1)}\right]$.
By \eqref{healphazy} and $h=\sum_{i}(h|v_i)v_i$,
\begin{align}\label{js1}
h_{(j)}\nu=h\delta_{j,1},\quad {\rm exp}\left(\sum_{j<0}\alpha_{(-j)}\frac{z^{j}}{j}\right)\nu=\nu-\alpha z^{-1}+\frac{1}{2}(\alpha|\alpha)z^{-2}.
\end{align}
So further by \eqref{field}
\begin{align}\label{js2}
e^{\alpha}_{(j)}\nu=&{\rm Res}_zz^{j}Y(e^{\alpha},z)\nu={\rm Res}_ze^{\alpha}z^{j} \sum_{n\geq0}p_n(\widetilde{\alpha})z^n\left(\nu-\alpha z^{-1}+\frac{1}{2}(\alpha|\alpha)z^{-2}\right)\notag\\
=&\frac{1}{2}(\alpha|\alpha)\delta_{j,1} e^{\alpha}+\frac{1}{2}\delta_{j,0} \left((\alpha|\alpha)-2\right)\alpha _{(-1)} e^{\alpha}.
\end{align}
Based upon above, we can get
\begin{align}\label{v1hm}
\left[h_{(m)},\nu_{(1)}\right]=mh_{(m)}, \quad\left[e^\alpha_{(m)},\nu_{(1)}\right] =\frac{1}{2}(\alpha|\alpha) m e^{\alpha}_{(m)}+\frac{1}{2}\left((\alpha|\alpha)-2\right)\left(\alpha _{(-1)} e^{\alpha}\right)_{(m+1)}.
\end{align}

According to \eqref{hh'} \eqref{alphaebeta} and \eqref{v1hm}, we can get the realization of $$\widehat{sl}_2 = \widetilde{sl}_2\oplus\mathbb{C}c\oplus\mathbb{C}d$$
on lattice vertex algebra $V_Q$, where $\widetilde{sl}_2 =\oplus_{k\in\mathbb{Z}}\lambda^ksl_2(\mathbb{C})$.
\begin{proposition}\label{sl2dddsbs}
$\widehat{sl}_2$ can be realized as the subalgebra of the lattice vertex algebra $V_Q\ (\mathcal{N}=2)$, i.e.
$$\widehat{sl}_2=\oplus_{n\in\mathbb{Z}}\mathbb{C}e_{(n)}^{v_1-v_2} \bigoplus\oplus_{n\in\mathbb{Z}}\mathbb{C}e_{(n)}^{v_2-v_1} \bigoplus\oplus_{n\in\mathbb{Z}}\mathbb{C}(v_1-v_2)_{(n)} \bigoplus\mathbb{C}\nu_{(1)}\bigoplus\mathbb{C}|0\rangle_{(-1)},$$
where $c$ is the center and $[d,h\lambda^m]=mh\lambda^m.$
Specifically, the correspondences are given by
\begin{align}\label{dddsbs}
H\lambda^n\mapsto(v_1-v_2)_{(n)},\quad E\lambda^n\mapsto\mathbf{i}(e^{v_1-v_2})_{(n)},\quad F\lambda^n\mapsto\mathbf{i}(e^{v_2-v_1})_{(n)},\quad c\mapsto|0\rangle_{(-1)},\quad d\mapsto -\nu_{(1)},
\end{align}
where $\mathbf{i}=\sqrt{-1}$, $H=\begin{pmatrix}
1&0\\
0&-1\\
\end{pmatrix},$\
$E=\begin{pmatrix}
0&1\\
0&0\\
\end{pmatrix},$\
$F=\begin{pmatrix}
0&0\\
1&0\\
\end{pmatrix}.$ Thus $sl_2(\mathbb{\mathbb{C}})$ can be realized in $V_Q$ by
\begin{align}\label{EFH}
H\mapsto v_1-v_2,\quad E\mapsto \mathbf{i}e^{v_1-v_2},\quad F\mapsto \mathbf{i}e^{v_2-v_1}.
\end{align}
\end{proposition}

\section{Untwisted $\widehat{sl}_2$--integrable hierarchy}
In this section, we will firstly construct the $\widehat{sl}_2$--integrable hierarchy by the Kac--Wakimoto construction from the homogenous representation of $\widehat{sl}_2$, which is called the untwisted $\widehat{sl}_2$--integrable hierarchy. Then by using the decomposition of the Casimir operator of $\widehat{sl}_2$, we get the equivalent DJKM construction for $\widehat{sl}_2$--integrable hierarchies. It is found that this
$\widehat{sl}_2$--integrable hierarchy is just the $1$--Toda lattice hierarchy.

\subsection{Kac--Wakimoto construction of untwisted $\widehat{sl}_2$--integrable hierarchy}\label{jsyy1}
If we introduce the realization of Heisenberg fields $v_i(z)$ in the way below $(n>0)$
\begin{align*}
v_{i(n)}=\partial_{x_{n}^{(i)}},\quad v_{i(-n)}=nx_{n}^{(i)},\quad v_{i(0)}=\partial_{v_i},\quad i=1,2,
\end{align*}
then we can get the homogeneous representation of $\widehat{sl}_2$ on
$$B=\mathbb{C}[x^{(1)},x^{(2)};e^{\pm v_1},e^{\pm v_2}],$$
by the untwisted realization of $\widehat{sl}_2$ in the lattice vertex algebra $V_Q$, where $x^{(i)}=
\left(x_1^{(i)},x_2^{(i)},\cdots\right)$. For convenience, let us denote this representation to be $\pi$. Note that $B$ is not irreducible, the irreducible representation is
$$B'=\mathbb{C}[t;q^{\pm1}],$$
where $t=(x^{(1)}-x^{(2)})/2=(t_1,t_2,\cdots),\ q=e^{v_1-v_2}$ and the corresponding realization of $v_i(z)$ is given by
\begin{align}\label{b'sdbs}
(v_1-v_2)_{(n)}=\partial_{t_n},\quad (v_1-v_2)_{(-n)}=2nt_n,\quad (v_1-v_2)_{(0)}=q\partial_q,\quad n>0.
\end{align}
In what follows, we denote the representation map corresponding to \eqref{b'sdbs} as $\pi'$.

There is a non-degenerate symmetric invariant bilinear form $(\cdot|\cdot)$ in $\widehat{sl}_2$ defined by
\begin{align*}
\left( A\lambda^k|B\lambda^l \right ) =k\delta_{k+l,0}{\rm tr}(A|B),\ \ \left(d|A\lambda^k\right)=0, \ \ \left(c|A\lambda^k\right)=0,\ \ (c|c)=0,\ \ (c|d)=1,\ \ (d|d)=0,\ \ \text{where}\ A\in sl_2(\mathbb{C}).
\end{align*}

\begin{lemma}\label{sl2doj}
$\{H\lambda^n,E\lambda^n, F\lambda^n,c,d\}$ and $\left\{\frac{1}{2} H\lambda^{-n},F\lambda^{-n}, E\lambda^{-n},d,c\right\}$ form the dual basis of $\widehat{sl}_2$.
\end{lemma}
\noindent Then the Casimir operator $\Omega_{KW}$ of $\widehat{sl}_2$ is
\begin{align*}
\Omega_{KW}= \sum_m\left(\frac{1}{2} H\lambda^{m}\otimes H\lambda^{-m}+E\lambda^m\otimes F\lambda^{-m}+F\lambda^m\otimes E\lambda^{-m}\right)+c\otimes d+d\otimes c,\quad m\in\mathbb{Z}.
\end{align*}
Therefore by Proposition \ref{sl2dddsbs}, $$\Omega_{KW}=\sum_m\left(\frac{1}{2}\big(v_{1}-v_{2}\big)_{(m)}\otimes\big(v_{1}-v_{2}\big)_{(-m)}-e^{v_1-v_2}_{(m)}\otimes e^{v_2-v_1}_{(-m)}-e^{v_2-v_1}_{(m)}\otimes e^{v_1-v_2}_{(-m)}\right)-|0\rangle_{(-1)}\otimes \nu_{(1)}-\nu_{(1)}\otimes |0\rangle_{(-1)},$$
which can be viewed as the $z^2$--coefficient of the Casimir field $Y(\Omega,z)$ with $\Omega$ given by
\begin{align}\label{Omegakw}
\Omega=\frac{1}{2}\big(v_{1}-v_{2}\big)\otimes\big(v_{1}-v_{2}\big)-e^{v_1-v_2}\otimes e^{v_2-v_1}-e^{v_2-v_1}\otimes e^{v_1-v_2}-|0\rangle\otimes \nu-\nu\otimes |0\rangle.
\end{align}
\begin{lemma}\label{omega1aa1}
For $a(w)=\sum_{n}a_{(n)}w^{-n-1}$ with $a\in\{v_1-v_2,e^{v_1-v_2},e^{v_2-v_1}\},$
\begin{align}\label{oa11a}
[Y(\Omega,z),a(w)\otimes1+1\otimes a(w)]=0.
\end{align}
\end{lemma}
\begin{proof}
Denote $A= a\otimes1+1\otimes a$, then according to
$$[A_{(m)},\Omega_{(n)}]=\sum_{j=0}^{\infty}\binom{m}{j}(A_{(j)}\Omega)_{(m+n-j)},$$
the key is to compute $A_{(j)}\Omega$.
Firstly by $Y(a\otimes b,z)=Y(a,z)\otimes Y(b,z)$, we can know
$$(a\otimes b)_{(n)}=\sum_{i}a_{(i)}\otimes b_{(n-i-1)}.$$
Then by \eqref{hh'} \eqref{ebetah} \eqref{js1} \eqref{js2} and
$$|0\rangle_{(n)}b=\delta_{n,-1}b,\quad b_{(n)}|0\rangle=\delta_{n,-1}b,\quad b\in V_Q,$$
we can obtain $A_{(j)}\Omega=0\ (j\geq0).$ Thus \eqref{oa11a} is correct.
\end{proof}

By the vacuum axiom of vertex algebra, we can find
$Y(\Omega,z)(|0\rangle\otimes|0\rangle)$ contains only the terms of $z\geq0,$ which means
\begin{align}\label{o00}
\Omega_{(n)}(|0\rangle\otimes|0\rangle)=0,\quad n\geq0.
\end{align}
\begin{proposition}\label{sxxfc}
If set $\tau={\rm exp}(a)|0\rangle$ with $a\in \widetilde{sl}_2$ expressed by \eqref{dddsbs}, then
\begin{align*}
\Omega_{(n)}(\tau\otimes\tau)=0,\quad n\geq 0.
\end{align*}
\end{proposition}
\begin{proof}
Firstly note that
$${\rm exp}(a)\otimes {\rm exp}(a)={\rm exp}(a\otimes1+1\otimes a),$$
then by \eqref{oa11a} and \eqref{o00}
\begin{align*}
\Omega_{(n)}(\tau\otimes\tau)=\Omega_{(n)}(e^a\otimes e^a)(|0\rangle\otimes|0\rangle)=(e^a\otimes e^a)\Omega_{(n)}(|0\rangle\otimes|0\rangle)=0.
\end{align*}
\end{proof}

Note that by $\Omega_{KW}=\Omega_{(1)}$ the Kac--Wakimoto construction is
$$\Omega_{(1)}(\tau\otimes\tau)=0.$$
Next we will try to realize $\Omega_{(1)}(\tau\otimes\tau)=0$ on the space $B'$. For this, the lemmas below are needed.

\begin{lemma}\label{1}
Given $f(t)\in\mathbb{C}[t],$
\begin{align*}
&Y(e^{v_1-v_2},z)\left(q^mf(t)\right)=z^{m}q^{m+1} e^{2\xi(t,z)}f\left(t-\left[z^{-1}\right]\right),\\
&Y(v_1-v_2,z)\left(q^mf(t)\right)=\sum_{n>0}q^m\partial_{t_n} f(t)z^{-n-1} -\sum_{n<0}2nt_{-n}q^mf(t)z^{-n-1}+mq^mf(t)z^{-1},
\end{align*}
where $\left[z^{-1}\right]=
\left(z^{-1},z^{-2}/2,z^{-3}/3,\cdots\right)$.
\end{lemma}
\begin{lemma}\label{2}
For $f(t)\in\mathbb{C}[t],$
\begin{align*}
\nu_{(1)}\left(q^mf(t)\right) =q^m\left(m^2+\sum_{j=1}^{+\infty}2jt_j\partial_{t_j}\right)f(t).
\end{align*}
\end{lemma}
\begin{proof}
By introducing
\begin{align}\label{u1u2}
u_1=\frac{v_1-v_2}{2},\quad u_2=\frac{v_1+v_2}{2},
\end{align}
we have
\begin{align*}
Y(\nu,z)=\frac{1}{2}\sum_{i=1,2}Y(v_{i(-1)}v_i,z) =Y\left(u_{1(-1)}u_1,z\right)+Y\left(u_{2(-1)}u_2,z\right).
\end{align*}
In particular on $B'$,
\begin{align*}
{\rm Res}_zzY(\nu,z) =\sum_{j\leq-1}u_{1(j)}u_{1(-j)}+\sum_{j\geq0}u_{1(-j)}u_{1(j)}=u_{1(0)}^2+2\sum_{j>0}u_{1(-j)}u_{1(j)}.
\end{align*}
Thus this lemma can be proved by \eqref{b'sdbs}.
\end{proof}
\begin{proposition}
Given the $\widehat{sl}_2$--integrable hierarchy
\begin{align*}
\Omega_{KW}(\tau\otimes\tau)=0,
\end{align*}
where $\tau$ is given in Proposition \ref{sxxfc}, if set $\pi'(\tau)=\sum_{m\in\mathbb{Z}}q^m\tau_m(t),$ then \eqref{oKW} can be written as \cite{KAC1989}
\begin{align}\label{oKW}
&{\rm Res}_zz^{s-\ell-1}e^{2\xi\left(t'-t'',z\right)}\tau_{s-1} \left(t'-\left[z^{-1}\right]\right) \tau_{\ell+1}\left(t''+\left[z^{-1}\right]\right) +{\rm Res}_zz^{\ell-s-1}e^{-2\xi\left(t'-t'',z\right)}\tau_{s+1} \left(t'+\left[z^{-1}\right]\right) \tau_{\ell-1}\left(t''-\left[z^{-1}\right]\right)\notag\\
=&\sum_{i>0}i\left(t'_i-t_i''\right) \left(\partial_{t'_i} -\partial_{t_i''}\right) \left(\tau_{s}(t') \tau_{\ell}(t'')\right) +\frac{1}{4}(s-\ell)^2\tau_{s}(t') \tau_{\ell}(t''),\quad s,l\in\mathbb{Z}.
\end{align}
\end{proposition}
\begin{proof}
If set $t'=t\otimes1,\ t''=1\otimes t,\ q'=q\otimes1,\ q''=1\otimes q,$ we can get by Lemma \ref{1} and Lemma \ref{2}
\begin{align}\label{r1}
&{\rm Res}_zz \left(Y(v_1-v_2,z)\otimes Y(v_1-v_2,z)\right)(\tau\otimes\tau)\notag\\
=&\sum_{s,\ell\in\mathbb{Z}}\left(2\sum_{i>0} it''_{i}\partial_{t'_i}\tau_s(t')\tau_{\ell}(t'') + it'_i\tau_s(t')\partial_{t''_{i}}\tau_{\ell}(t'')+ s\ell\tau_s(t')\tau_{\ell}(t'')\right)q^{\prime s}q^{\prime\prime\ell},\notag\\
&{\rm Res}_zzY\left(e^{v_1-v_2},z\right)\otimes Y\left(e^{v_2-v_1},z\right)(\tau\otimes\tau)\notag\\ =&\sum_{s,\ell\in\mathbb{Z}}{\rm Res}_zq^{\prime s}q^{\prime\prime\ell} z^{s-\ell-1}e^{2\xi(t'-t'',z)} \tau_{s-1}\left(t'-\left[z^{-1}\right]\right) \tau_{\ell+1}\left(t''+\left[z^{-1}\right]\right),
\end{align}
and
\begin{align}\label{r2}
&{\rm Res}_zz\left(1\otimes Y(\nu,z)+ Y(\nu,z)\otimes1\right)(\tau\otimes\tau)\notag\\
=&\sum_{s,\ell\in\mathbb{Z}}\left( \sum_{i>0}it_i'' \tau_s(t')\partial_{t''_{i}}\tau_{\ell}(t'') +\sum_{i>0}it'_i \partial_{t'_{i}}\tau_s(t')\tau_{\ell}(t'') +\frac{1}{4}\ell^2 \tau_s(t')\tau_{\ell}(t'') +\frac{1}{4} s^2 \tau_s(t')\tau_{\ell}(t'')\right)q^{\prime s}q^{\prime\prime\ell}.
\end{align}
Substitute \eqref{r1} and \eqref{r2} into
$\pi'\otimes\pi'\left(\Omega_{KW}(\tau\otimes\tau)\right)=0$ and take coefficients of $q^{\prime s}q^{\prime\prime\ell}$, then we can prove \eqref{oKW}.
\end{proof}
\subsection{Equivalent DJKM construction of untwisted $\widehat{sl}_2$--integrable hierarchy}
In order to consider the decomposition of $\Omega_{KW}$, we need to introduce
\begin{align*}
\Omega^+=e^{v_1}\otimes e^{-v_1}+e^{v_2}\otimes e^{-v_2},\quad
\Omega^-=e^{-v_1}\otimes e^{v_1}+e^{-v_2}\otimes e^{v_2}.
\end{align*}
The relation of $\Omega^{\pm}$ and $\Omega$ is given in the lemma below.
\begin{lemma}\label{5+7}
$\Omega$ is related with $\Omega^{\pm}$ by
\begin{align*}
2\left(\Omega+\Omega^{+}_{(-1)}\Omega^{-}\right) =-(v_1+v_2)\otimes(v_1+v_2)
+(v_1+v_2)_{(-2)}|0\rangle\otimes|0\rangle-|0\rangle\otimes (v_1+v_2)_{(-2)}|0\rangle.
\end{align*}
\end{lemma}
\begin{proof}
Firstly note that
\begin{align*}
\Omega^+_{(-1)}\Omega^-=&{\rm Res}_zz^{-1}Y(\Omega^+,z)\Omega^-=\sum_{i,j=1,2}{\rm Res}_zz^{-1}Y(e^{v_i},z)e^{-v_j}\otimes Y(e^{-v_i},z)e^{v_j}\\
=&\sum_{i,j=1,2}\sum_{m\in\mathbb{Z}}(e^{v_i})_{(m)}e^{-v_j}\otimes (e^{-v_i})_{(-m-2)}e^{v_j}.
\end{align*}
Then by \eqref{heee},
\begin{align*}
(e^{v_i})_{(m)}e^{-v_j}\otimes (e^{-v_i})_{(-m-2)}e^{v_j} =p_{\delta_{ij}-m-1}(\widetilde{v_i})e^{v_i-v_j}\otimes p_{\delta_{ij}+m+1}(-\widetilde{v_i})e^{-v_i+v_j}.
\end{align*}
Therefore
\begin{align}\label{omegaomega}
&\Omega^+_{(-1)}\Omega^-=\sum_{i=1,2}\sum_{m=-2}^0p_{-m}(\widetilde{v_i})|0\rangle \otimes p_{m+2}(-\widetilde{v_i})|0\rangle+\sum_{i\neq j}e^{v_i-v_j}\otimes e^{-v_i+v_j}\notag\\
=&e^{v_1-v_2}\otimes e^{v_2-v_1}+e^{v_2-v_1}\otimes e^{v_1-v_2}+\left(\nu\otimes1+1\otimes\nu\right)\notag\\
 &+\frac{1}{2}\sum_{i=1}^2\left(v_{i(-2)}|0\rangle\otimes|0\rangle-|0\rangle\otimes v_{i(-2)}|0\rangle\right)-v_{1}\otimes v_{1}-v_{2}\otimes v_{2},
\end{align}
where we have used $v_{i(-1)}|0\rangle=v_i$. Finally by comparing \eqref{Omegakw} and \eqref{omegaomega}, we can prove this lemma.
\end{proof}
\begin{corollary}\label{Omeganfj}
On $B'$, $\Omega_{(n)}$ can be expressed by
$$\Omega_{(n)}=-\Omega_{(n-1)}^-\Omega^+_{(0)}-\sum_{k>0} \left(\Omega^+_{(-k)}\Omega_{(n-1+k)}^-+\Omega_{(n-1-k)}^-\Omega^+_{(k)}\right).$$
\end{corollary}
\begin{proof}
Firstly let us recall that
\begin{align*}
Y(a,z)_{(-1)}Y(b,z)=:Y(a,z)Y(b,z):,
\end{align*}
where
$$:Y(a,z)Y(b,z):=Y(a,z)_+Y(b,z)+p(a,b)Y(b,z)Y(a,z)_-,$$
and $Y(a,z)_+=\sum_{n<0}a_{(n)}z^{-n-1},\ Y(a,z)_-=\sum_{n\geq0}a_{(n)}z^{-n-1}.$
Note that on $B'$, $(v_1+v_2)_{(n)}=0$. Therefore
\begin{align*}
\Omega(z)=-:\Omega^+(z)\Omega^-(z):=-\sum_{k,\ell}:\Omega^+_{(k)} \Omega^-_{(\ell)}:z^{-k-\ell-2}=-\sum_{k,n}:\Omega^+_{(k)}\Omega^-_{(n-1-k)}:z^{-n-1},
\end{align*}
By comparing the coefficients of $z^{-n-1}$, we can get the result of this corollary.
\end{proof}

If set $q_1=e^{v_1},\ q_2=e^{v_2}$, then we can define the Hermitian form on $B$ as follows \cite{KAC2013}
\begin{small}
$$H\left(P_1\left(q_1,q_2,x^{(1)},x^{(2)}\right), P_2\left(q_1,q_2,x^{(1)},x^{(2)}\right)\right)={\rm Res}_{q_1}{\rm Res}_{q_2}q_1^{-1}q_2^{-1} P_1\left(q_1^{-1},q_2^{-1},\widetilde{\partial}_{x^{(1)}}, \widetilde{\partial}_{x^{(2)}}\right) \overline{P_2\left(q_1,q_2,x^{(1)},x^{(2)}\right)} \big|_{x^{(1)}=x^{(2)}=0}.$$
\end{small}
where $\widetilde{\partial}_{x^{(i)}}=\left(\partial_{x^{(i)}_1}, \partial_{x^{(i)}_2}/2, \partial_{x^{(i)}_3}/3,\cdots\right)$, $\overline{P_2}$ is the complex conjugation of $P_2$. In particular
\begin{align}\label{emnjybjg}
H\left( q_1^{m'_1}q_2^{m'_2}x^{(1)\alpha'}x^{(2)\beta'} ,q_1^{m''_1}q_2^{m''_2}x^{(1)\alpha''}x^{(2)\beta''}\right) =\delta_{m'_1,m''_1}\delta_{m'_2,m''_2}\delta_{\alpha',\alpha''} \delta_{\beta',\beta''}\alpha'!\beta'! \prod_{i=1}^{+\infty}\frac{1}{i^{\alpha'_i+\beta'_i}},
\end{align}
where $\eta=(\eta_1,\eta_2,\eta_3,\cdots)$ for $\eta\in\{\alpha',\beta',\alpha'',\beta''\}$ and $\delta_{\alpha',\alpha''}=\prod_{i=1}^{+\infty}\delta_{\alpha'_i,\alpha''_i}$.
Therefore
\begin{align}\label{ge}
(e^{\gamma})^{\dag}=e^{-\gamma},\quad \gamma_{(i)}^{\dag}=\gamma_{(-i)},
\end{align}
where $\gamma\in\mathbb{Z}v_1\oplus\mathbb{Z}v_2$ and $A^{\dag}$ is the adjoint operator of $A$ determined by
$$H( A(P_1),P_2)=H\left(P_1,A^{\dag}(P_2)\right).$$ It is obvious by \eqref{emnjybjg} that this Hermitian form is positive.
\begin{lemma}\label{yegdag}
For $\gamma\in Q,$
$$Y\left(e^{\gamma},z\right)^{\dag}=z^{-1}Y\left(e^{-\gamma},z^{-1}\right).$$
\end{lemma}
\begin{proof}
According to \eqref{ge} and $(AB)^{\dag}=B^{\dag}A^{\dag}$,
\begin{align*}
Y\left(e^{\gamma},z\right)^{\dag}=&{\rm exp}\left(\sum_{n<0}\gamma_{(n)}\frac{z^{n}}{n}\right) {\rm exp}\left(\sum_{n>0}\gamma_{(n)}\frac{z^{n}}{n}\right)z^{\gamma _{(0)}}e^{-\gamma}\\
=&z^{-1}e^{-\gamma}z^{\gamma_{(0)}}{\rm exp}\left(\sum_{n<0}\gamma_{(n)}\frac{z^{n}}{n}\right) {\rm exp}\left(\sum_{n>0}\gamma_{(n)}\frac{z^{n}}{n}\right)=z^{-1}Y\left(e^{-\gamma},z^{-1}\right),
\end{align*}
where we have used \eqref{healpha}.
\end{proof}
We can extend the above Hermitian form to the space $B\otimes B$ by
$$H(f_1\otimes g_1,f_2\otimes g_2) =H(f_1,f_2) H( g_1,g_2),$$
where $f_i,g_i\in B$. One can check that
\begin{align}\label{abdag0}
(P_1\otimes P_2)^{\dag}=P_1^{\dag}\otimes P_2^{\dag}.
\end{align}
\begin{lemma}\label{omegapm}
$\left(\Omega^+_{(m)}\right)^{\dag}=\Omega^-_{(-m)}.$
\end{lemma}
\begin{proof}
Firstly $Y(\Omega^+,z)=\sum_{i=1}^2Y(e^{v_i},z)\otimes Y(e^{-v_i},z)$. Then by Lemma \ref{yegdag} and \eqref{abdag0},we can get
$$Y(\Omega^+,z)^{\dag}=\sum_{i=1}^2z^{-2}Y\left(e^{-v_i},z^{-1}\right)\otimes Y\left(e^{v_i},z^{-1}\right).$$
Finally by comparing the coefficient of $z^{-m-1}$, we can get $\left(\Omega^+_{(m)}\right)^{\dag}=\Omega^-_{(-m)}.$
\end{proof}

\begin{theorem}\label{12}
For $\tau\in \mathbb{C}[t]$,
\begin{align}\label{zhyy3}
\Omega_{KW}(\tau\otimes\tau) =0\Longleftrightarrow\Omega^{\pm}_{(m)}(\tau\otimes\tau)=0, \quad  m\geq0.
\end{align}
\end{theorem}
\begin{proof}
Since $\Omega_{KW}=\Omega_{(1)}$, thus by Corollary \ref{Omeganfj} we can get
\begin{align}\label{okw}
\Omega_{KW}=-\Omega^{+}_{(0)}\Omega^{-}_{(0)} -\sum_{m=1}^{+\infty}\left(\Omega^-_{(-m)}\Omega^{+}_{(m)} +\Omega^{+}_{(-m)}\Omega^-_{(m)}\right).
\end{align}
By \eqref{okw}, it can be found that $\Omega^+_{(m)}(\tau\otimes\tau)=0$ implies $\Omega_{KW}(\tau\otimes\tau)=0$. Conversely when $\Omega_{\rm KW}(\tau\otimes\tau)=0$,
\begin{align*}
0=&H\left(\tau\otimes\tau,-\Omega_{\rm KW}(\tau\otimes\tau)\right)\\
=&H\left(\tau\otimes\tau,\Omega^-_0\Omega^+_0 (\tau\otimes\tau)\right) +\sum_{m=1}^{+\infty}H\left(\tau\otimes\tau, \Omega^-_{(-m)}\Omega^+_{(m)} (\tau\otimes\tau)\right)
+H\left(\tau\otimes\tau,\Omega^+_{(-m)}\Omega^-_{(m)} (\tau\otimes\tau)\right)\\
=&H\left(\Omega_0^+(\tau\otimes\tau),\Omega_0^+(\tau\otimes\tau)\right) +\sum_{m=1}^{+\infty}H\left(\Omega^+_{(m)}(\tau\otimes\tau) ,\Omega^+_{(m)}(\tau\otimes\tau)\right) +H\left(\Omega^-_{(m)}(\tau\otimes\tau) ,\Omega^-_{(m)}(\tau\otimes\tau)\right),
\end{align*}
where we have used \eqref{okw}.
From the positive definiteness of \ $H(\cdot,\cdot)$, we can finally obtain
\begin{align*}
\Omega_{(m)}^+(\tau\otimes\tau)=\Omega^-_{(m)}(\tau\otimes\tau)=0,\ \quad m\geq0.
\end{align*}
\end{proof}

\begin{remark}
If set
$$Y(e^{\pm v_i},z)=\sum_{k\in\mathbb{Z}+1/2}\psi^{\pm (i)}_kz^{-k-1/2}=\psi^{\pm (i)}(z),$$
then
$$\psi^{\pm(i)}_k\psi^{\pm(j)}_{\ell} +\psi^{\pm(j)}_{\ell}\psi^{\pm(i)}_k=0,\quad \psi^{\pm(i)}_k\psi^{\mp(j)}_{\ell} +\psi^{\mp(j)}_{\ell}\psi^{\pm(i)}_k=\delta_{k+\ell,0}\delta_{k,\ell},$$
which can be proved by \eqref{alphaebeta}. Therefore $Y(e^{\pm v_i},z)$ is just the $2$--component charged free fermion and
$$\Omega_{(m)}^{\pm}={\rm Res}_z\sum_{i=1}^2 z^m\psi^{\pm(i)}(z)\otimes \psi^{\mp(i)}(z)=\sum_{j\in\mathbb{Z}+1/2}\left(\psi_j^{\pm (1)}\otimes \psi_{m-j}^{\mp (1)}+\psi_j^{\pm (2)}\otimes \psi_{m-j}^{\mp (2)}\right),$$
which shows $\Omega_{(m)}^{\pm}(\tau\otimes\tau)=0$ is the equivalent DJKM construction \cite{JIMBO} of $\widehat{sl}_2$--integrable hierarchy.

\end{remark}

\begin{remark}\label{remark}
If $\tau\in\mathbb{C}[[t]]$, we can not use the same proof as Theorem \ref{12}, since the Hermitian form $H(\cdot,\cdot)$ does not extend to $\mathbb{C}[[t]]$. But \eqref{zhyy3} is still correct for $\tau\in\mathbb{C}[[t]]$. Actually by Lemma \ref{2}
$$\nu_{(1)}=2\sum_{j=1}^{+\infty}jt_j\partial_{t_j}\quad\text{on} \ \ \mathbb{C}[[t]].$$
Then it can be find that $\mathbb{C}[[t]]=\oplus_{n\geq0}V_n,$
where
$$V_n=\{a\in\mathbb{C}[[t]]|\nu_{(1)}a=na\}$$
has finite dimension. For $\tau\in\mathbb{C}[[t]]$ if assume $$\tau\otimes\tau=\sum_{n\geq0}^{+\infty}a_n,$$
satisfies $\Omega_{KW}(\tau\otimes\tau)=0,$
where
$$\left(\nu_{(1)}\otimes1+1\otimes\nu_{(1)}\right)(a_n)=na_n,\quad  a_n\in\mathbb{C}[t]\otimes\mathbb{C}[t].$$
Then by the fact $\nu_{(1)}\otimes1+1\otimes\nu_{(1)}$ commutes with $\Omega_{KW}$ on $\mathbb{C}[[t]]\otimes\mathbb{C}[[t]]$, $\Omega_{KW}(a_n)$ has the same degree with $a_n$ and thus
$$\Omega_{KW}(a_n)=0,\quad n\geq0.$$
So by similar proof in Theorem \ref{12}, we can get  $\Omega^{\pm}_{(m)}(a_n)=0$, and further $\Omega^{\pm}_{(m)}(\tau\otimes\tau)=0.$
\end{remark}

Next we will try to rewrite $\Omega^{\pm}_{(m)}(\tau\otimes\tau)=0$ into the usual form. For this, let us rewrite $Y(e^{\pm v_i},z)$ in the way below
\begin{align*}
Y(e^{\pm v_1},z)&=e^{\pm v_1}z^{\pm u_{1(0)}}{\rm exp}\left(\sum_{n<0}\mp u_{1(n)}\frac{z^{-n}}{n}\right){\rm exp}\left(\sum_{n>0}\mp u_{1(n)}\frac{z^{-n}}{n}\right),\\
Y(e^{\pm v_2},z)&=e^{\pm v_2}z^{\mp u_{1(0)}}{\rm exp}\left(\sum_{n<0}\pm u_{1(n)}\frac{z^{-n}}{n}\right){\rm exp}\left(\sum_{n>0}\pm u_{1(n)}\frac{z^{-n}}{n}\right),
\end{align*}
where $u_i$ is given by \eqref{u1u2} and we have set $u_{2(n)}=0$, then apply $Y(e^{\pm v_i},z)$ on $q^mf(t)$,
\begin{align*}
Y(e^{\pm v_1},z)\left(q^mf(t)\right)&=e^{\pm v_1}q^mz^{\pm m}e^{\pm\xi(t,z)} f\left(t\mp \left[z^{-1}\right]/2\right),\\
Y(e^{\pm v_2},z)\left(q^mf(t)\right)&=e^{\pm v_2}q^mz^{\mp m}e^{\mp\xi(t,z)} f\left(t\pm \left[z^{-1}\right]/2\right).
\end{align*}
After inserting above relation into
$$\Omega^{\pm}_{(m)}(\tau\otimes\tau)={\rm Res}_zz^m\left(Y(e^{v_1},z)\otimes Y(e^{-v_1},z) +Y(e^{v_2},z)\otimes Y(e^{-v_2},z)\right)(\tau\otimes\tau)=0,$$
with $\tau=\sum_{\ell}q^{\ell}\tau_{\ell}(t)$ the following applying $e^{-v_1}\otimes e^{v_1}$ on both sides, we can get
\begin{align}\label{eflqxdsxxfc}
&\oint_{C_{\infty}}\frac{dz}{2\pi \mathbf{i}}z^{m+\ell-\ell^{\prime}}e^{\xi\left(t-t^{\prime},z\right)} \tau_{\ell}\left(t-\left[z^{-1}\right]/2\right) \tau_{\ell^{\prime}}\left(t^{\prime}+\left[z^{-1}\right]/2 \right)\notag\\
=&\oint_{C_{0}}\frac{dz}{2\pi \mathbf{i}}z^{-m+\ell-\ell^{\prime}}e^{\xi\left(t^{\prime}-t,z^{-1}\right)} \tau_{\ell+1}\left(t+[z]/2\right)\tau_{\ell^{\prime}-1} \left(t^{\prime}-[z]/2\right),
\end{align}
where $C_{\infty}$ means the circle around $z=\infty$, while $C_{0}$ means the circle around $z=0$. Both $C_{\infty}$ and $C_{0}$ are anticlockwise. Here we used $\oint_{C_{\infty}}\frac{dz}{2\pi \mathbf{i}}f(z)={\rm Res}_zf(z)=\oint_{C_{0}}\frac{dz}{2\pi \mathbf{i}}\frac{1}{z^2}f(z)$.

Note that \eqref{eflqxdsxxfc} is equivalent to
the following Lax equation \cite{Takasaki2018}
$$L_{t_n}=\left[A_n,L\right],$$
where the Lax operator $L$ is given by $$L=\Lambda+b_{\ell}\left(t\right) +c_{\ell}\left(t\right)\Lambda^{-1},\quad A_n=\frac{1}{2}L^n_{\geq0}-\frac{1}{2}L^n_{<0}.$$
Here $\Lambda(f_{\ell}(t))=f_{\ell+1}(t)$ and
$$b_{\ell}(t)=\partial_{t_1}{\rm log}\frac{\tau_{\ell+1}(t)} {\tau_{\ell}(t)},\ c_{\ell}(t)=\frac{\tau_{\ell+1}(t)\tau_{\ell-1}(t)}{\tau_{\ell}^2(t)}.$$
\section{Twisted $\widehat{sl}_2$--integrble hierarchy}

In this section, we will firstly review some basic facts about the twisted module over lattice vertex algebras. One can refer to \cite{BAKALOV} for more details. Then based upon this, we will give the twisted representation of $\widehat{sl}_2$ and compute the corresponding twisted $\widehat{sl}_2$--integrable hierarchy. Finally, we will give the equivalent DJKM construction of twisted $\widehat{sl}_2$--integrable hierarchy by decomposition of Casimir operator. It is found that the twisted $\widehat{sl}_2$--integrable hierarchy is just the KdV hierarchy.
\subsection{Twisted module over lattice vertex algebra}
Assume $V$ is a vertex algebra, $\sigma$ is an automorphism on $V$ with finite order $N$. Denote
$$V_j=\left\{a\in V|\sigma(a)=\epsilon^{-j}a\right\},\ \ \ \ \epsilon=e^{2\pi \mathbf{i}/N},\ \   0\leq j\leq N-1.$$
Note that if $a\in V$, then
\begin{align}\label{pij}
\pi_j(a)=\frac{1}{N}\Sigma_{k=0}^{N-1}\epsilon^{kj}\sigma^k(a)
\end{align}
belongs to $V_j$.
A $\sigma$--twisted $V$--module is a vector space $M$ with a linear map from $V$ to the space of $N$--twisted fields on $M$,
$$a\longmapsto Y^{M}(a,z)=\Sigma_{n\in\mathbb{Z}/N}a^M_{(n)}z^{-n-1},\ \ \ a^M_{(n)}\in {\rm End}(M),$$
which satisfies for $a,b\in V,\ c\in M$
\begin{align}\label{ytsigma}
Y^{M}(\sigma (a),z)=Y^{M}\left(a,e^{2\pi \mathbf{i}}z\right),\quad Y^{M}(|0\rangle,z)={\rm Id}_M,
\end{align}
and Borcherds identity for $a\in V_j,\ b\in V,\ c\in M,\ m\in j/N+\mathbb Z,\ N\in\mathbb Z,\ k\in \mathbb{Z}/N$:
\begin{align}
\sum_{i\geq0}
                             \binom{m}{i}
                           \left(a_{(n+i)}^Mb^M\right)_{(m+k-i)}c^M =\sum_{i\geq0}(-1)^i
                             \binom{n}{i}
                           a_{(m+n-i)}^M\left(b^M_{(k+i)}c^M\right) -\sum_{i\geq0}(-1)^{i+n}
                             \binom{n}{i}
                           b^M_{(n+k-i)}\left(a_{(m+i)}^Mc^M\right).\label{borcher}
\end{align}

By comparing the coefficients of $(\ref{ytsigma})$, we can get
for $Y^M(a,z)=\Sigma_{n\in\mathbb{Z}/N}a^M_{(n)}z^{-n-1}$ with $a\in V_j$,
$$a^M_{(n)}=0\ \ \text{for}\ \ Nn\not\equiv j\ {\rm mod}\ N .$$
Therefore
\begin{align}\label{ymaz}
Y^M(a,z)=\Sigma_{n\in\mathbb Z}a^M_{(n+j/N)}z^{-n-j/N-1}.
\end{align}
The Borcherds identity $(\ref{borcher})$ implies for $a\in V_j$
\begin{align}
 \left[a^M_{(m)},b^M_{(n)}\right]=\sum_{k=0}^\infty \left(
                             \begin{array}{c}
                              m \\
                               k \\
                             \end{array}
                           \right)\left(a_{(k)}b\right)^M_{(m+n-k)},\ \ \ \ m,n\in\mathbb{Z}/N,\label{anbn}
\end{align}
and
\begin{align}
Y^M(a,w)_{(k)}Y^M(b,w)=\sum_{l=0}^\infty \left(
                             \begin{array}{c}
                              \frac{j}{N} \\
                               l \\
                             \end{array}
                           \right)w^{-l}Y^M\left(a_{(l+k)}b,w\right),\ \ \ \ k\in\mathbb{Z}/N\label{ytnyt},
\end{align}
where
\begin{align*}
Y^M(a,w)_{(k)}Y^M(b,w)={\rm Res}_zz^{\frac{j}{N}}w^{\frac{-j}{N}}\left(Y^M(a,z)Y^M(b,w) i_{z,w}(z-w)^k
- p(a,b)Y^M(b,w)Y^M(a,z)i_{w,z}(z-w)^k\right).
\end{align*}
In particular
\begin{align}\label{ymym}
Y^M(a,w)_{(-1)}Y^M(b,w)=:Y^M(a,w)Y^M(b,w):,\quad a\in V_j.
\end{align}
Here for $m,n\in\mathbb{Z}/N$,
\begin{align}\label{zhyy2}
:a^M_{(m)}b^M_{(n)}:=\left\{
    \begin{array}{ll}
      \ \ \ \ \ a^M_{(m)}b^M_{(n)}, &  m<0,\\
      p(a,b)b^M_{(n)}a^M_{(m)}, &  m\geq0.
    \end{array}
  \right.
\end{align}
One should note that \eqref{zhyy2} is slightly different from \eqref{zhyy} in $:a^M_{k/N}b^M_{l/N}:$ for $-N<k<0$.

Now let us see the twisted module over lattice vertex algebra $V_Q$ (see Section $1$).
Let $\sigma$ be an automorphism of lattice $Q$. Then $\sigma$ can be lifted to the automorphism of the lattice vertex algebra $V_Q$ by
\begin{align}
\sigma\left(h\lambda^m\right)=\sigma(h)\lambda^m,\ \ \ \ \sigma\left(e^{\alpha}\right)=\eta(\alpha)^{-1}e^{\sigma(\alpha)},\quad \alpha \in Q,\label{sigmah}
\end{align}
where $\eta:\ Q\longrightarrow\{\pm1\}$ satisfies
\begin{align}
\eta(\alpha)\eta(\beta)\varepsilon(\alpha,\beta) =\eta(\alpha+\beta)\varepsilon(\sigma(\alpha),\sigma(\beta)), \quad \alpha,\ \beta \in Q.\label{etaalpha}
\end{align}
Moreover $\eta$ is chosen such that
\begin{align}\label{etaalpha1}
\eta(\alpha)=1,\quad \text{if}\ \sigma(\alpha)=\alpha\quad \text{for}\ \alpha\in Q.
\end{align}
In particular, $\sigma$ fixes the conformal vector $\nu$, that is $\sigma(\nu)=\nu.$
\begin{remark}
By setting $\alpha=\beta=0$ and $\beta=-\alpha$ in \eqref{etaalpha} respectively, we have the following results
\begin{align}\label{eta0=-}
\eta(0)=1,\quad \eta(\alpha)\eta(-\alpha) =\frac{\varepsilon(\sigma(\alpha),\sigma(\alpha))}{\varepsilon(\alpha,\alpha)}.
\end{align}
\end{remark}
In what follows, we assume the automorphism $\sigma$ given by $(\ref{sigmah})$ have order $N$.
In the twisted module $M$, the field $Y^{M}(e^\alpha,z)\ (\alpha\in Q)$ is given by
\begin{align}\label{Ytwisted1}
Y^{M}(e^\alpha,z)=z^{b_{\alpha}}U_\alpha^ME_\alpha^M(z),
\end{align}
where
\begin{align}\label{2star}
b_\alpha=\frac{|\pi_0(\alpha)|^2-|\alpha|^2}{2},\quad E^M_\alpha(z)=z^{\alpha_{(0)}^M}{\rm exp}\left(\sum_{n\in\mathbb Z_{<0}/N}\alpha_{(n)}^M\frac{z^{-n}}{-n}\right){\rm exp}\left(\sum_{n\in\mathbb Z_{>0}/N}\alpha_{(n)}^M\frac{z^{-n}}{-n}\right),
\end{align}
and $U_\alpha^M$ is the operator on $M$ satisfying
\begin{align}\label{Ytwisted2}
\left[h_{(m)}^M,U_\alpha^M\right]=\delta_{m,0}(\pi_0(h)|\alpha)U_\alpha^M, \ \ \ \ \ U_{\sigma(\alpha)}^M=\eta(\alpha)U_{\alpha}^Me^{2\pi \mathbf{i}\left(b_{\alpha}+\alpha^M_{(0)}\right)},
\end{align}
with $ h\in \mathfrak{h},\ m\in \mathbb Z/N.$ For $\alpha_1,\alpha_2\in Q$
\begin{align}\label{3star} U_{\alpha_1}^MU_{\alpha_2}^M=\varepsilon(\alpha_1,\alpha_2)B^{-1}_{\alpha_1,\alpha_2}U^M_{\alpha_1+\alpha_2}, \quad B_{\alpha_1,\alpha_2}=N^{-(\alpha_1|\alpha_2)}\Pi_{k=1}^{N-1} \left(1-e^{\frac{2\pi \mathbf{i}k}{N}}\right)^{\left(\sigma^k(\alpha_1)|\alpha_2\right)}.
\end{align}
By $(\ref{anbn})$, we can get
\begin{align*}
\left[h'^M_{(m)},(h'')^M_{(n)}\right]=(\pi_{Nm}(h')|h'') m\delta_{m,-n}{\rm Id}_M,\quad \left[h^M_{(m)},Y^M(e^{\alpha},w)\right]=(\pi_{Nm}(h)|\alpha)w^mY^M(e^{\alpha},w),
\end{align*}
where $h,h',h''\in\mathfrak{h},\ m,n\in\mathbb{Z}/N,\ \alpha\in Q.$

\subsection{Twisted representation of $\widehat{sl}_2$}
In this subsection, we will consider the $\sigma$--twised module over the lattice vertex algebra $V_Q$ with $Q=\mathbb{Z}v_1\oplus\mathbb{Z}v_2$. Here the automorphism $\sigma$ of $V_Q$ is given by \eqref{sigmah} and
$$\sigma(v_1)=v_2,\ \sigma(v_2)=v_1.$$
\begin{lemma}
For $m,n\in\mathbb{Z}$,
\begin{align*}
\eta(mv_1+nv_2)\eta(mv_2+nv_1)=(-1)^{m^2-n^2},\quad \eta(mv_1+nv_2)\eta(-mv_1-nv_2)=1.
\end{align*}
\end{lemma}
\begin{proof}
Firstly by \eqref{varepsilon}, \eqref{etaalpha} and \eqref{etaalpha1} one can find
\begin{align*}
\eta(mv_1+nv_2)\eta(mv_2+nv_1) &=\frac{\varepsilon(mv_2+nv_1,mv_1+nv_2)} {\varepsilon(mv_1+nv_2,mv_2+nv_1)} \eta\left((m+n)(v_1+v_2)\right)\\
&=\frac{\varepsilon(v_2,v_1)^{m^2} \varepsilon(v_1,v_1)^{nm} \varepsilon(v_2,v_2)^{mn} \varepsilon(v_1,v_2)^{n^2}}{\varepsilon(v_1,v_2)^{m^2} \varepsilon(v_1,v_1)^{mn} \varepsilon(v_2,v_2)^{nm} \varepsilon(v_2,v_1)^{n^2}}=(-1)^{m^2-n^2},
\end{align*}
Similarly according to \eqref{eta0=-}
\begin{align*}
\eta(mv_1+nv_2)\eta(-mv_1-nv_2) =\frac{\varepsilon(mv_2+nv_1,mv_2+nv_1)} {\varepsilon(mv_1+nv_2,mv_1+nv_2)}
=\frac{\varepsilon(v_2,v_2)^{m^2} \varepsilon(v_2,v_1)^{mn} \varepsilon(v_1,v_2)^{mn} \varepsilon(v_1,v_1)^{n^2}}{\varepsilon(v_1,v_1)^{m^2} \varepsilon(v_1,v_2)^{mn} \varepsilon(v_2,v_1)^{mn} \varepsilon(v_2,v_2)^{n^2}}=1.
\end{align*}
\end{proof}
\noindent By this lemma, the order of $\sigma$ on $V_Q$ is $N=4$. And in particular
$$\eta(v_1)\eta(v_2)=-1,\quad \eta(v_1)=\eta(-v_1),\quad \eta(v_2)=\eta(-v_2),$$
so in what follows, we set
$$\eta(v_1)=1,\ \eta(v_2)=-1.$$
Moveover, by \eqref{EFH} we can find
$$\sigma(H)=-H,\quad \sigma(E)=F,\quad \sigma(F)=E.$$
where we have used
$$\eta(v_1-v_2)=\frac{\varepsilon(v_1,-v_2)\eta(v_1)\eta(-v_2)} {\varepsilon(v_2,-v_1)} =-\eta(v_1)\eta(-v_2)=1.$$
Moreover it can be found that $\sigma$ acts on $sl_2$ by \cite{Kroode}
\begin{align}\label{YS}
\sigma={\rm exp}\left(\frac{\pi \mathbf{i}}{2}{\rm ad} (E+F)\right),
\end{align}
so the order of $\sigma$ on $sl_2$ is $2$.

Let us recall the $\sigma$--twisted realization $\widehat{L}(sl_2,\sigma)$ of $\widehat{sl}_2$, which is the subalgebra of $\widehat{sl}_2$ given by
$$\widehat{L}(sl_2,\sigma)=\oplus_{j\in\mathbb{Z}} \lambda^jsl_2(\mathbb{C})_{\bar{j}}\oplus\mathbb{C}c \oplus\mathbb{C}d.$$
Here
$$sl_2(\mathbb{C})_{\bar{j}}=\{a\in sl_2|\sigma(a)=(-1)^{-j}a\}.$$
From \eqref{YS}, $\sigma$ is the inner automorphism of $sl_2$, thus according to \cite{Kac1990} (See Chapter $8$)
\begin{align}\label{tg}
\widehat{L}(sl_2,\sigma) \cong\widehat{sl}_2.
\end{align}

In terms of lattice vertex algebra $V_Q$, we can find
$$\widehat{L}(sl_2,\sigma) =\oplus_{n\in\mathbb{Z}} \mathbb{C}(v_1-v_2)^{tw}_{(n)} \oplus\mathbb{C} \left(e^{v_1-v_2}\right) ^{tw}_{(n)} \oplus\mathbb{C}\nu_{(1)} \oplus\mathbb{C}|0\rangle_{(-1)},$$
where for $x\in V_Q$, we set $x^{tw}_{(n)}=(\pi_n(x))_{(n)}$ with $N=2$ and $\epsilon=-1$ in \eqref{pij}. Note that if $\sigma^k(x)=y$ for $x,y\in V_Q$, then $$\mathbb{C}x^{tw}_{(n)}=\mathbb{C}y^{tw}_{(n)}.$$

\noindent It can be found that
\begin{align*}
&\left(e^{v_1-v_2}\right)^{tw}_{(2n+1)} =\frac{1}{2}\left(\left(e^{v_1-v_2}\right)_{(2n+1)} -\left(e^{v_2-v_1}\right)_{(2n+1)}\right),\quad (v_1-v_2)^{tw}_{(2n)}=0,\\
&\left(e^{v_1-v_2}\right)^{tw}_{(2n)} =\frac{1}{2}\left(\left(e^{v_1-v_2}\right)_{(2n)} +\left(e^{v_2-v_1}\right)_{(2n)}\right),\quad (v_1-v_2)^{tw}_{(2n+1)}=(v_1-v_2)_{(2n+1)}.
\end{align*}
The corresponding Lie brackets are given by
\begin{align}\label{con}
&\left[(v_1-v_2)^{tw}_{(2m+1)},(v_1-v_2)^{tw}_{(2n+1)}\right] =2(2m+1)\delta_{m+n+1,0}|0\rangle_{(-1)},\quad \left[\left(e^{v_1-v_2}\right)^{tw}_{(2m)}, \left(e^{v_1-v_2}\right)^{tw}_{(2n)}\right] =-m\delta_{m+n,0}|0\rangle_{(-1)},\notag\\
&\left[\left(e^{v_1-v_2}\right)^{tw}_{(2m+1)}, \left(e^{v_1-v_2}\right)^{tw}_{(2n+1)}\right] =\left(m+\frac{1}{2}\right)\delta_{m+n+1,0}|0\rangle_{(-1)},\quad \ \ \left[(v_1-v_2)^{tw}_{(2m+1)},\left(e^{v_1-v_2}\right)^{tw}_{(n)}\right] =2\left(e^{v_1-v_2}\right)^{tw}_{(2m+n+1)},\notag\\
&\left[\left(e^{v_1-v_2}\right)^{tw}_{(2m)}, \left(e^{v_1-v_2}\right)^{tw}_{(2n+1)}\right] =\frac{1}{2}(v_1-v_2)^{tw}_{(2m+2n+1)}.
\end{align}
Thus if define the following linear map \cite{Kroode}
$$\Phi:\ \widehat{L}(sl_2,\sigma)\rightarrow \widehat{sl}_2$$
by
\begin{align*}
&\Phi\left(\left(e^{v_1-v_2}\right)^{tw}_{(2m)}\right) =-\frac{\mathbf{i}}{2}\left(v_1-v_2\right)_{(m)} +\frac{\mathbf{i}}{2}\delta_{m,0}|0\rangle_{(-1)},\quad \Phi\left(\left(e^{v_1-v_2}\right)^{tw}_{(2m+1)}\right) =\frac{\mathbf{i}}{2}\left(e^{v_1-v_2}\right)_{(m)}+\frac{\mathbf{i}}{2}\left(e^{v_2-v_1}\right)_{(m+1)},\\
&\Phi\left((v_1-v_2)^{tw}_{(2m+1)}\right) =\left(e^{v_1-v_2}\right)_{(m)}-\left(e^{v_2-v_1}\right)_{(m+1)},\quad \ \ \Phi\left(|0\rangle_{(-1)}\right)=\frac{1}{2}|0\rangle_{(-1)},\quad \ \ \Phi(\nu_{(1)})=2\nu_{(1)},
\end{align*}
then it can be checked that $\Phi$ is an isomorphism, which confirms $\widehat{L}(sl_2,\sigma) \cong\widehat{sl}_2$ again.

In the construction of the twisted representation of $\widehat{sl}_2$, the $\sigma$--twisted Heisenberg algebra  $$\widehat{\eta}_{\sigma} =\oplus_{n\in\mathbb{Z}} \left(\mathbb{C}(v_1+v_2) _{(n)}^{\sigma} \oplus\mathbb{C}(v_1-v_2) _{(n+1/2)}^{\sigma}\right) \oplus\mathbb{C}I$$
plays a key role and the corresponding Lie brackets are given by
\begin{align}\label{1star}
&\left[(v_1+v_2)_{(m)}^{\sigma} ,(v_1+v_2)_{(n)}^{\sigma}\right]=2m\delta_{m+n,0}I,\quad \left[(v_1+v_2)_{(m)}^{\sigma} ,(v_1-v_2)_{(n+1/2)}^{\sigma}\right]=0,\notag\\
&\left[(v_1-v_2)_{(m+1/2)}^{\sigma} ,(v_1-v_2)_{(n+1/2)}^{\sigma}\right] =(2m+1)\delta_{m+n+1,0}I,\quad m,n\in\mathbb{Z}.
\end{align}
Next let us consider the highest weight module $\widehat{\eta}_{\sigma}$--module $M$ spanned by $$(v_1+v_2)_{(-m_1)}^{\sigma} \cdots(v_1+v_2)_{(-m_i)}^{\sigma}\ (v_1-v_2)_{(-n_1+1/2)}^{\sigma} \cdots (v_1-v_2)_{(-n_j+1/2)}^{\sigma}u^k \mathbf{1},$$
where $m_i,n_i\in\mathbb{Z}_{>0},\ k\in\mathbb{Z}.$ Here $\mathbf{1}\in M$ is the highest weight vector satisfying \begin{align*}
(v_1+v_2)^{\sigma}_{(m)} \mathbf{1}=(v_1-v_2)^{\sigma}_{(m+1/2)}\mathbf{1}=0, \quad m\in\mathbb{Z}_{\geq0},
\end{align*}
while $u$ is a linear operator on $M$ obeying  $$\left[(v_1+v_2)_{(m)}^{\sigma},u\right]=\delta_{m,0}u,\quad \left[(v_1-v_2)^{\sigma}_{(m+1/2)},u\right]=0,\quad m\in\mathbb{Z}.$$
On $M$, the operator $I$ acts as the identity operator ${\rm Id_M}.$

Denote $(v_1+v_2)_{(m)}^M$ and $(v_1-v_2)_{(m+1/2)}^M$ to be the linear operators on $M$ induced by $(v_1+v_2)^{\sigma}_{(m)}$ and $(v_1-v_2)^{\sigma}_{(m+1/2)}$ on $M$ respectively. And let $U_{v_1}^M$ be the linear operator on $M$ given by the action of $u$ on $M$. Set $$(v_1+v_2)^M_{(m+1/2)}=(v_1-v_2)^M_{(m)}=0,\quad (v_1\pm v_2)^M_{(m+1/4)}=(v_1\pm v_2)^M_{(m+3/4)}=0,\quad U_{v_2}^M=U_{v_1}^Me^{2\pi \mathbf{i}\left(-1/4+v_{1(0)}^M\right)},\quad m\in\mathbb{Z}.$$
Then it can be checked that
\begin{align*}
Y^M(\gamma,z)=\sum_{n\in\mathbb{Z}/4}\gamma^M_{(n)}z^{-n-1},\quad
Y^M(e^\gamma,z)=z^{b_\gamma}U^M_\gamma E^M_r(z),\quad \gamma\in Q,
\end{align*}
will satisfy the $\sigma$--invariance \eqref{ytsigma}, where $U^M_\gamma$ is computed by \eqref{3star}. Further by
$$Y^M(|0\rangle, z)={\rm Id_M},$$
we can get the action of $V_Q$ on $M$.

Before further discussion, we list some useful examples of $b_{\gamma}$ and $B_{\gamma,\gamma'},\ \gamma,\gamma'\in Q$,
$$b_{\pm v_1}=b_{\pm v_2}=-\frac{1}{4},\quad b_{\pm (v_1-v_2)}=-1,\quad B_{v_1,-v_2}=B_{v_2,-v_1}=\frac{1}{2},\quad B_{v_1,-v_1}=B_{v_2,-v_2}=2.$$
\begin{lemma}\label{Ynu}
$Y^M(\nu,z)=\frac{1}{4}:Y^M\left(v_1-v_2,z\right)Y^M\left(v_1-v_2,z\right): +\frac{1}{4}:Y^M\left(v_1+v_2,z\right)Y^M\left(v_1+v_2,z\right): +\frac{1}{16}z^{-2}{\rm Id}_M.$
\end{lemma}
\begin{proof}
Firstly the conformal vector
\begin{align}\label{x}
\nu=\frac{1}{2}\sum_{i=1}^2v_{i(-1)}v_{i(-1)}|0\rangle =\frac{1}{4}(v_1+v_2)_{(-1)}(v_1+v_2)+\frac{1}{4}(v_1-v_2)_{(-1)}(v_1-v_2).
\end{align}
Then by \eqref{ytnyt} \eqref{ymym},
\begin{align}\label{xx}
Y^M\left((v_1+v_2)_{(-1)}(v_1+v_2),z\right) =:Y^M(v_1+v_2,z)Y^M(v_1+v_2,z):
\end{align}
and
$$:Y^M(v_1-v_2,z)Y^M(v_1-v_2,z): =\sum_{l=0}^{+\infty}\binom{\frac{1}{2}}{l}z^{-l} Y^M\left((v_1-v_2)_{(l-1)}(v_1-v_2),z\right),$$
where one should note that $\sigma(v_1\pm v_2)=\pm(v_1\pm v_2)$. By \eqref{hh'} \eqref{hh'nj}, one can find that
$$(v_1-v_2)_{(l-1)}(v_1-v_2)=\delta_{l,0} (v_1-v_2)_{(-1)}(v_1-v_2)+2\delta_{l,2}|0\rangle,\quad l\geq0.$$
Therefore
\begin{align}\label{xxx}
Y^M\left((v_1-v_2)_{(-1)}(v_1-v_2),z\right) =:Y^M(v_1-v_2,z)Y^M(v_1-v_2,z):+\frac{z^{-2}}{4}{\rm Id}_M.
\end{align}
Finally this lemma can be proved by \eqref{x} \eqref{xx} \eqref{xxx}.
\end{proof}
Then we have the following proposition.
\begin{proposition}
The representation of $\widehat{L}(sl_2,\sigma)$ on $M$ is given by
\begin{align*}
&(v_1-v_2)^{tw}_{(n)}\rightarrow (v_1-v_2)^M_{(n/2)},\quad \left(e^{v_1-v_2}\right)^{tw}_{(n)}\rightarrow \left(e^{v_1-v_2}\right)^{M}_{(n/2)},\notag\\
&\left(e^{v_2-v_1}\right)^{tw}_{(n)}\rightarrow \left(e^{v_2-v_1}\right)^{M}_{(n/2)},\quad \ \ \ \nu_{(1)}\rightarrow 2\nu^M_{(1)},\quad \ \ \ |0\rangle_{(-1)}\rightarrow 2^{-1}{\rm Id}_M.
\end{align*}
\end{proposition}
\begin{proof}
It can be checked by \eqref{anbn} and \eqref{con}.
\end{proof}
\begin{remark}
By \eqref{ytsigma} \eqref{tg} and \eqref{con}, we can get the twisted representation of $\widehat{sl}_2$, which is the principal representation.
\end{remark}
\begin{lemma}\label{ya11a}
For $Y^M(a,w)=\sum_{n\in\mathbb{Z}/2}a^M_{(n)}w^{-n-1}$ with $a\in\{v_1-v_2,e^{v_1-v_2},e^{v_2-v_1}\}$,
$$\left[Y^M(\Omega,z),Y^M(a,w)\otimes1+1\otimes Y^M(a,w)\right]=0.$$
\end{lemma}
\begin{proof}
By \eqref{anbn} it can be proved by similar way as Lemma \ref{omega1aa1}.
\end{proof}

\subsection{Kac--Wakimoto construction of the twisted $\widehat{sl}_2$--integrable hierarchy}\label{js-2}
In the construction of the twisted $\widehat{sl}_2$--integrable hierarchy, we need to identify $M$ with $\widetilde{B}=\mathbb{C}[t,u^{\pm1}]$ by setting $$(v_1+v_2)_{(-m_1)}\cdots(v_1+v_2) _{(-m_i)}\ (v_1-v_2)_{(-n_1+1/2)}\cdots (v_1-v_2)_{(-n_j+1/2)}u^k\mathbf{1}$$
to be
$$m_1\cdots m_i(n_1-1/2)\cdots (n_j-1/2)t_{m_1}\cdots t_{m_i}\ t_{2n_1-1}\cdots t_{2n_j-1}u^k.$$
For $j>0$, the actions of $\widehat{\eta}_{\sigma}$ on $M$ are given by
\begin{align}\label{tqxbs}
&(v_1-v_2)_{(j-1/2)}^M=\partial_{t_{2j-1}},\quad (v_1-v_2)_{(-j+1/2)}^M=(2j-1)t_{2j-1},\notag\\
&(v_1+v_2)_{(j)}^M=\partial_{t_{2j}},\quad (v_1+v_2)_{(-j)}^M=2jt_{2j},\quad (v_1+v_2)_{(0)}^M=u\partial_u,
\end{align}
which obey \eqref{1star}. Note that $M$ is not the irreducible representation of $\widehat{sl}_2$. In fact if denote $$\widetilde{B}'=\mathbb{C}[u^{\pm},\hat{t}],$$
where $\hat{t}=(t_1,t_3,t_5,\cdots)$, then $\widetilde{B}'$ is the irreducible one.
\begin{lemma}\label{Yumft}
Given $u^mf(\hat{t})\in\widetilde{B}'$,
\begin{align*}
&Y^M(v_1-v_2,z)(u^mf(\hat{t})) =\left(\sum_{n\geq0}(2n+1)t_{2n+1}z^{n} +\partial_{t_{2n+1}}z^{-n-1} \right)z^{-1/2} f(\hat{t})u^m ,\notag\\
&Y^M(e^{\pm(v_1-v_2)},z)(u^mf(\hat{t})) =\frac{\mathbf{i}}{4}(-1)^{m}z^{m-1} e^{\pm2\hat{\xi}\left(\hat{t},z^{1/2}\right)} f\left(\hat{t}\mp 2\left[z^{-1/2}\right]_{\rm o}\right)u^m,\notag\\
&{\rm Res}_zzY^M(\nu,z)(u^mf(\hat{t})) =\frac{1}{2}\sum_{n\geq0}(2n+1)t_{2n+1} \partial_{t_{2n+1}}f(\hat{t})u^m+ \left(\frac{1}{4}m^2+\frac{1}{16}\right)f(\hat{t})u^m,
\end{align*}
where $[\lambda]_{\rm o}=(\lambda,\lambda^3/3,\lambda^5/5,\cdots),\ \hat{\xi}(\hat{t},\lambda)=\sum_{n\geq0}^{+\infty}t_{2n+1}\lambda^{2n+1}.$
\end{lemma}
\begin{proof}
Firstly $Y^M(v_1-v_2,z)(u^mf(\hat{t}))$ can be easily obtained by \eqref{tqxbs}. As for $Y^M(e^{\pm(v_1-v_2)},z)$, we can find by \eqref{Ytwisted1} \eqref{2star} \eqref{tqxbs}
$$Y^M(e^{v_1-v_2},z)(u^mf(\hat{t}))=z^{m-1}U_{v_1-v_2}^M e^{2\hat{\xi}(\hat{t},z^{1/2})} f\left(\hat{t}- 2\left[z^{-1/2}\right]_{\rm o}\right)u^m.$$
Next let us compute $U_{v_1-v_2}^M.$ For this by \eqref{Ytwisted2} and \eqref{3star}
\begin{align*}
U_{v_1-v_2}^M=\frac{1}{2}U_{v_1}^MU_{-v_2}^M =\frac{1}{2}U_{v_1}^MU_{-v_1}^Me^{2\pi \mathbf{i}\left(b_{-v_1}-v_{1(0)}^M\right)}=\frac{\mathbf{i}}{4}(-1)^{u\partial_u},
\end{align*}
where
$v^M_{1(0)} =\frac{1}{2}(v_1+v_2)_{(0)}^{M}+\frac{1}{2}(v_1-v_2)_{(0)}^M =\frac{1}{2}(v_1+v_2)_{(0)}^{M}=\frac{1}{2}u\partial_u$.
Then we have
$$Y^M\left(e^{v_1-v_2},z\right)(u^mf(\hat{t}))=\frac{\mathbf{i}}{4}(-1)^{m}z^{m-1} e^{2\hat{\xi}\left(\hat{t},z^{1/2}\right)} f\left(\hat{t}- 2\left[z^{-1/2}\right]_{\rm o}\right)u^m.$$
Similarly one can obtain the result of $Y^M(e^{-(v_1-v_2)},z)(u^mf(\hat{t}))$.

Finally let us see ${\rm Res}_zzY^M(\nu,z)(u^mf(\hat{t})).$ Notice that by Lemma \ref{Ynu} and \eqref{tqxbs}
\begin{align*}
\nu_{(1)}={\rm Res}_zzY^M(\nu,z) =\frac{1}{4}\sum_{n\in\mathbb{Z}} \left(:(v_{1}+v_{2})_{(n)}^M(v_{1}+v_{2})_{(-n)}^M: +:(v_{1}-v_{2})_{(n+1/2)}^M(v_{1}+v_{2})_{(-n-1/2)}^M: \right) +\frac{1}{16}.
\end{align*}
Therefore
\begin{align*}
{\rm Res}_zzY^M(\nu,z)(u^mf(\hat{t}))= \frac{1}{2}\sum_{n\geq0}(2n+1)t_{2n+1} \partial_{t_{2n+1}}f(\hat{t})u^m+\left( \frac{1}{4}m^2+\frac{1}{16}\right)f(\hat{t})u^m .
\end{align*}
\end{proof}
By Lemma \ref{Yumft} and \eqref{Omegakw}, we can get the lemma below.
\begin{lemma}\label{lemma}
$\Omega^M_{(1)}(\mathbf{1}\otimes\mathbf{1}) =0.$
\end{lemma}
By Lemma \ref{ya11a}, Lemma \ref{lemma} and
$${\rm Res}_zz^{-1}f(z)={\rm Res}_zz^{-1}f(-z)= {\rm Res}_zz^{-1}f\left(z^{1/2}\right),$$
we have the following result.
\begin{proposition}\label{js}
If set $\tau={\rm exp}(a)|0\rangle$ with $a\in \widetilde{sl}_2$, then
\begin{align*}
\Omega_{(1)}^M(\tau\otimes\tau)=0,
\end{align*}
which is the following twisted $\widehat{sl}_2$--integrable hierarchy \cite{KAC1989}
\begin{align}
&{\rm Res}_zz^{-1}e^{2\hat{\xi}(\hat{t}'-\hat{t}'',z)} \tau\left(\hat{t}'-2\left[z^{-1}\right]_{\rm o}\right)
\tau\left(\hat{t}''+2\left[z^{-1}\right]_{\rm o}\right)\nonumber \\
=&4\sum_{j\geq0}(2j+1)\left(\hat{t}'_{2j+1}-\hat{t}''_{2j+1}\right) \left(\partial_{\hat{t}'_{2j+1}}
-\partial_{{\hat{t}''_{2j+1}}}\right)
\left(\tau(\hat{t}')\tau(\hat{t}'')\right) +\tau(\hat{t}')\tau(\hat{t}'').\label{KdVKW}
\end{align}
\end{proposition}

\subsection{Decomposition of Casimir operator and the $\widehat{sl}_2$--integrable hierarchy}\ \\
\indent Recall that $\Omega^{\pm}=\sum_{i=1}^2e^{\pm v_i}\otimes e^{\mp v_i}$, so denote
\begin{align}\label{yozofschs}
Y^M\left(\Omega^{\pm},z\right) =\sum_{m\in\mathbb{Z}/4}\Omega^{\pm M}_{(m)}z^{-m-1} =\sum_{i=1}^2Y^M(e^{\pm v_i},z)\otimes Y^M(e^{\mp v_i},z).
\end{align}
Then by \eqref{sigmah} \eqref{eta0=-} and
$\sigma(\Omega^{\pm})=\sigma\left(e^{\pm v_1}\right)\otimes\sigma\left(e^{\mp v_1}\right) +\sigma\left(e^{\pm v_2}\right)\otimes\sigma\left(e^{\mp v_2}\right)$, we can get
\begin{align}\label{oo}
\sigma\left(\Omega^{\pm}\right)=\Omega^{\pm }.
\end{align}
Therefore by \eqref{ymaz}
$$Y^M(\Omega^{\pm},z)=\sum_{m\in\mathbb{Z}}\Omega^{\pm M}_{(m)}z^{-m-1}.$$
\begin{proposition}\label{okwt}
On $\widetilde{B}'$
\begin{align*}
\Omega_{(1)}^M=-\Omega^{+M}_{(0)}\Omega^{-M}_{(0)} -\sum_{m=1}^{+\infty}\left(\Omega^{-M}_{(-m)}\Omega^{+M}_{(m)} +\Omega^{+M}_{(-m)}\Omega^{-M}_{(m)}\right).
\end{align*}
\end{proposition}
\begin{proof}
By \eqref{ytnyt} and \eqref{oo}, we can find on $\widetilde{B}'$
$$Y^M\left(\Omega^{+}_{(-1)} \Omega^{-},z\right) =Y^M\left(\Omega^{+},z\right)_{(-1)} Y^M\left(\Omega^{-},z\right) =:Y^M\left(\Omega^{+},z\right) Y^M\left(\Omega^{-},z\right):.$$
Then this proposition can be proved by comparing coefficients of $z^{-2}$.
\end{proof}

We can define the positive definite Hermitian form on $\widetilde{B}$ as follows \cite{KAC2013}
$$\widetilde{H}\left(P_1(u,t), P_2(u,t)\right)={\rm Res}_{u}u^{-1} P_1\left(u^{-1},\widetilde{\partial}_{t}\right) \overline{P_2(u,t)}\big|_{t=0},$$
where $\widetilde{\partial}_{t}=(\partial_{t_1}, \partial_{t_2}/2,\partial_{t_3}/3\cdots,)$. Similarly we can extend the above Hermitian form to the space $\widetilde{B}\otimes \widetilde{B}$ by $$\widetilde{H}(f_1\otimes g_1,f_2\otimes g_2) =\widetilde{H}(f_1,f_2) \widetilde{H}( g_1,g_2),$$
where $f_i,g_i\in \widetilde{B}$. One can check that
\begin{align}\label{abdag}
(P_1\otimes P_2)^{\dag}=P_1^{\dag}\otimes P_2^{\dag}.
\end{align}
Therefore
\begin{align}\label{tge}
u^{\dag}=u^{-1},\quad \gamma_{(i)}^{M\dag}=\gamma_{(-i)}^M,\quad \gamma\in Q,\quad i\in\mathbb{Z}/4.
\end{align}
It is obvious by \eqref{emnjybjg} that this Hermitian form is positive definite.
\begin{lemma}\label{23}
For $i=1,2$,
$
Y^{M}(e^{\pm v_i},z)^{\dag} =-2^{\pm1}z^{-1}Y^{M}(e^{\mp v_i},z^{-1}).$
\end{lemma}
\begin{proof}
For $\gamma\in \{v_1,v_2\}$, according to \eqref{2star}
$$\left(E_{\gamma}^M(z)\right)^{\dag}={\rm exp}\left(\sum_{n\in\mathbb Z_{<0}/4}\gamma_{(n)}^M\frac{z^{n}}{n}\right) {\rm exp}\left(\sum_{n\in\mathbb Z_{>0}/4}\gamma_{(n)}^M\frac{z^{n}}{n}\right)z^{\gamma_{(0)}^M} =E_{-\gamma}^M(z^{-1}).$$
Further by \eqref{Ytwisted2}
$$ U_{\gamma}^Mz^{\gamma_{(0)}^M} =z^{\gamma_{(0)}^M-1/2}U_{\gamma}^M,\quad U_{\gamma}^MU_{-\gamma}^M=-\frac{1}{2},\quad U^M_{\gamma}e^{2\pi \mathbf{i}\gamma^M_{(0)}}=-e^{2\pi \mathbf{i}\gamma^M_{(0)}}U_\gamma^M,\quad U_{\pm v_2}^M=-\mathbf{i}U_{\pm v_1}^Me^{\pm 2\pi \mathbf{i}v_{1(0)}^M}.$$
So by \eqref{Ytwisted2} and \eqref{tge},
\begin{align*}
\left(U_{v_1}^M\right)^{\dag}=\left(U_{v_1}^{M}\right)^{-1} =-2U_{-v_1}^M,\quad \left(U_{v_2}^{M}\right)^{\dag} =\mathbf{i}e^{-2\pi \mathbf{i}v_{1(0)}^M}\left(U_{v_1}^{M}\right)^{-1}=-2\mathbf{i}e^{-2\pi \mathbf{i}v_{1(0)}^M}U_{-v_1}^M=-2U_{-v_2}^M,
\end{align*}
where we have used
$(AB)^{\dag}=B^{\dag}A^{\dag}$ and $(\lambda A)^{\dag}=\bar{\lambda}A^{\dag}$. Therefore according to \eqref{Ytwisted1}
\begin{align}\label{yy}
Y^{M}(e^{\gamma},z)^{\dag}=-2z^{-3/4} U_{-\gamma}^ME_{-\gamma}^M(z^{-1}) = -2z^{-1}Y^{M}(e^{-\gamma},z^{-1}),
\end{align}
Further by \eqref{yy},
$$Y^{M}(e^{-\gamma},z^{-1}) =-\frac{1}{2}zY^{M}(e^{\gamma},z)^{\dag},$$
that is
$$Y^{M}(e^{-\gamma},z)^{\dag} =-\frac{1}{2}z^{-1}Y^{M}(e^{\gamma},z^{-1}).$$
\end{proof}
\begin{lemma}\label{omegapm}
$\left(\Omega^{+M}_{(m)}\right)^{\dag}=\Omega^{-M}_{(-m)}.$
\end{lemma}
\begin{proof}
Firstly by \eqref{yozofschs} \eqref{abdag} and Lemma \ref{23}
$$Y^M(\Omega^+,z)^{\dag}=\sum_{i=1}^2z^{-2} Y^M(e^{-v_i},z^{-1})\otimes Y^M(e^{v_i},z^{-1}).$$
Then by comparing the coefficients of $z^{-m-1}$, we can get $\left(\Omega^{+M}_{(m)}\right)^{\dag}=\Omega^{-M}_{(-m)}.$
\end{proof}

\begin{theorem}\label{theorem:KdV}
For $\tau\in \mathbb{C}[\hat{t}]$,
\begin{align}\label{zhyy4}
\Omega_{(1)}^M(\tau\otimes\tau) =0\Longleftrightarrow\Omega^{\pm M}_{(m)}(\tau\otimes\tau) =0, \quad  m\geq0.
\end{align}
\end{theorem}
\begin{proof}
By Proposition \ref{okwt}, it can be found that $\Omega^{+M}_{(m)}(\tau\otimes\tau)=0$ implies $\Omega_{(1)}^M(\tau\otimes\tau)=0$. Conversely when $\Omega_{(1)}^M(\tau\otimes\tau)=0$,
\begin{align*}
0=&\widetilde{H}\left(\tau\otimes\tau,-\Omega^M_{(1)}(\tau\otimes\tau)\right)\\
=&\widetilde{H}\left(\tau\otimes\tau,\Omega^{-M}_{(0)}\Omega^{+M}_{(0)} (\tau\otimes\tau)\right) +\sum_{m=1}^{+\infty}\widetilde{H}\left(\tau\otimes\tau, \Omega^{-M}_{(-m)}\Omega^{+M}_{(m)} (\tau\otimes\tau)\right)
+\widetilde{H}\left(\tau\otimes\tau,\Omega^{+M}_{(-m)}\Omega^{-M}_{(m)} (\tau\otimes\tau)\right)\\
=&\widetilde{H}\left(\Omega_{(0)}^{+M}(\tau\otimes\tau), \Omega_{(0)}^{+M}(\tau\otimes\tau) \right) +\sum_{m=1}^{+\infty}\widetilde{H}\left(\Omega^{+M}_{(m)}(\tau\otimes\tau) ,\Omega^{+M}_{(m)}(\tau\otimes\tau) \right) +\widetilde{H}\left(\Omega^{-M}_{(m)}(\tau\otimes\tau) ,\Omega^{-M}_{(m)}(\tau\otimes\tau)\right),
\end{align*}
where we have used Proposition \ref{okwt} and Lemma \ref{omegapm}.
From the positive definiteness of \ $\widetilde{H}(\cdot,\cdot)$, we can finally prove this theorem.
\end{proof}

\begin{remark}
Firstly note that
$$\sigma\left(e^{\pm v_1}-\mathbf{i}e^{\pm v_2}\right)=\epsilon^{-3}\left(e^{\pm v_1}-\mathbf{i}e^{\pm v_2}\right),\quad \sigma\left(e^{\pm v_1}+\mathbf{i}e^{\pm v_2}\right)=\epsilon^{-1}\left(e^{\pm v_1}+\mathbf{i}e^{\pm v_2}\right).$$
where $\epsilon={\rm exp}(2\pi \mathbf{i}/4)=\mathbf{i}$. Thus according to \eqref{ymaz},
\begin{align*}
Y^M\left(e^{\pm v_1}-\mathbf{i}e^{\pm v_2} ,z\right)=\sum_{n\in\mathbb{Z}}\left(e^{\pm v_1}-\mathbf{i}e^{\pm v_2}\right)^M_{(n+3/4)}z^{-n-3/4-1},\ \  Y^M\left(e^{\pm v_1}+\mathbf{i}e^{\pm v_2} ,z\right)=\sum_{n\in\mathbb{Z}}\left(e^{\pm v_1}+\mathbf{i}e^{\pm v_2}\right)^M_{(n+1/4)}z^{-n-1/4-1},
\end{align*}
If set
\begin{align*}
&\left(e^{v_1}-\mathbf{i}e^{v_2}\right)^M_{(n-1/4)}=2\psi^+_{2n+1/2}, \quad \ \ \ \ \left(e^{v_1}+\mathbf{i}e^{v_2}\right)^M_{(n+1/4)}=2\psi^+_{2n+1+1/2}, \\ &\left(e^{-v_1}-\mathbf{i}e^{-v_2}\right)^M_{(n-1/4)}=-\psi^-_{2n+1/2}, \quad \left(e^{-v_1}-\mathbf{i}e^{-v_2}\right)^M_{(n+1/4)}=-\psi^-_{2n+1+1/2}, \quad
\end{align*}
then by \eqref{anbn}
$$\psi^{\pm}_{k+1/2}\psi^{\pm}_{l+1/2} +\psi^{\pm}_{l+1/2}\psi^{\pm}_{k+1/2}=0,\quad \psi^{\pm}_{k+1/2}\psi^{\mp}_{l+1/2} +\psi^{\mp}_{l+1/2}\psi^{\pm}_{k++1/2}=\delta_{k+l+1,0},$$
which shows that $\psi^{\pm}_{k+1/2}$ are charged free fermions \cite{JIMBO,Miwa2000}. Furthermore if let $\psi^{\pm}(z)=\sum_{k\in\mathbb{Z}} \psi^{\pm}_{k+1/2}z^{-k-1}$, then
\begin{align*}
&Y^M\left(e^{v_1},z\right) =\psi^+\left(z^{\frac{1}{2}}\right)z^{-1/4} ,\quad \ \ \ \ \ \ Y^M\left(e^{-v_1},z\right) =-\frac{1}{2}\psi^-\left(z^{\frac{1}{2}}\right)z^{-1/4}\\  &Y^M\left(e^{v_2},z\right) =-\mathbf{i}\psi^+\left(-z^{\frac{1}{2}}\right)z^{-1/4} ,\quad Y^M\left(e^{-v_2},z\right) =\frac{\mathbf{i}}{2}\psi^-\left(-z^{\frac{1}{2}}\right)z^{-1/4} ,
\end{align*}
So we have
\begin{align*}
\Omega_{(m)}^{\pm M}=\frac{1}{2}{\rm Res}_zz^{m-\frac{1}{2}}\sum_{i=1}^2(-1)^{i}\psi^{\pm}\left((-1)^{i+1}z^{\frac{1}{2}}\right)\otimes \psi^{\mp}\left((-1)^{i+1}z^{\frac{1}{2}}\right) =-\sum_{j\in\mathbb{Z}+1/2}\psi^{\pm}_j\otimes\psi^{\mp}_{2m-j}.
\end{align*}
Therefore $\Omega_{(m)}^{\pm M}(\tau\otimes\tau)=0$ is the equivalent DJKM construction \cite{JIMBO,Miwa2000} of twisted $\widehat{sl}_2$--integrable hierarchy.
\end{remark}

Next we will try to rewrite $\Omega_{(m)}^{\pm M}(\tau\otimes\tau)=0$ into the usual form.
\begin{remark}
If $\tau\in\mathbb{C}[[\hat{t}]],$ \eqref{zhyy4} is still correct. Please refer to Remark \ref{remark}.
\end{remark}
\begin{lemma}
For $n\in\mathbb{Z}_{>0}$,
\begin{align*}
&v_{1(n/2)}^M=\frac{1}{2}\partial_{t_n},\quad  v_{1(-n/2)}^M=\frac{1}{2}nt_n,\quad v_{2(n/2)}^M=\frac{1}{2}(-1)^n\partial_{t_n},\quad v_{2(-n/2)}^M=\frac{1}{2}(-1)^nnt_n,\\
&v_{1(0)}^M=\frac{1}{2}u\partial_u,\quad\   v_{2(0)}^M=\frac{1}{2}u\partial_u,\quad\ \ \ \
v^M_{1(n/2+1/4)}=v^M_{2(n/2+1/4)}=0.
\end{align*}
\end{lemma}
\begin{proof}
It can be proved by \eqref{tqxbs} and for $n\in\mathbb{Z}$,
$$v_{1(n/2)}^M=\frac{1}{2}(v_1+v_2)^M_{(n/2)} +\frac{1}{2}(v_1-v_2)^M_{(n/2)},\quad v_{2(n/2)}^M=\frac{1}{2}(v_1+v_2)^M_{(n/2)} -\frac{1}{2}(v_1-v_2)^M_{(n/2)}.$$
\end{proof}

\begin{lemma}\label{Yeviz}
For $u^mf(\hat{t})\in\widetilde{B}',\ j=1,2$,
\begin{align*}
&Y^M(e^{v_j},z)\left(u^mf(\hat{t})\right) =(-\mathbf{i})^{j-1}(-1)^{(j-1)m}z^{\frac{m}{2}-\frac{1}{4}} e^{\xi\left(t,(-1)^{j-1}z^{1/2}\right)} f\left(\hat{t}+(-1)^j\left[z^{-1/2}\right]_{\rm o}\right)u^{m+1}, \\
&Y^M(e^{- v_j},z)\left(u^mf(\hat{t})\right) =-\frac{1}{2}(-\mathbf{i})^{j-1}(-1)^{(j-1)m}z^{-\frac{m}{2}-\frac{1}{4}} e^{-\xi\left(t,(-1)^{j-1} z^{1/2}\right)} f\left(\hat{t}+(-1)^{j-1}\left[z^{-1/2}\right]_{\rm o}\right)u^{m-1} .
\end{align*}
\end{lemma}
\begin{proof}
Firstly by \eqref{Ytwisted1}
\begin{align*}
E^M_{\pm v_i}(z) =& z^{\pm v_{i(0)}^M}{\rm exp}\left(\pm\sum_{n>0}2v_{i(-n/2)}^M\frac{z^{n/2}}{n}\right) {\rm exp}\left(\mp\sum_{n>0}2v_{i(n/2)}^M\frac{z^{-n/2}}{n}\right) \\
=&z^{\pm u\partial_u/2}{\rm exp}\left(\pm\xi\left(t,(-1)^{i-1}z^{1/2}\right)\right) {\rm exp}\left(\mp\xi\left(\widetilde{\partial_t},(-1)^{i-1}z^{-1/2}\right)\right).
\end{align*}
Then according to the definitions of $U_{v_i}^M,$ we can get
$$U_{v_1}^M=u,\quad U_{v_2}^M=-\mathbf{i}U_{v_1}^M{\rm exp}\left(2\pi \mathbf{i}v_{1(0)}^M\right)=-\mathbf{i}u(-1)^{u\partial_u}.$$
Next by $U_{v_i}^MU_{-v_i}^M=-\frac{1}{2},$
$$U_{-v_1}^M=-\frac{1}{2}u^{-1},\quad U^M_{-v_2}=\frac{\mathbf{i}}{2}u^{-1}(-1)^{-u\partial_u}.$$
Then the final result of this lemma can be obtained by  \eqref{Ytwisted1} acting on $u^mf(\hat{t})$.
\end{proof}
\begin{remark}
Note that there is the dependence on $t_{2n}$ in the exponential part of $Y^M(e^{\pm v_i},z)(u^mf(\hat{t}))$.
\end{remark}
By Lemma \ref{Yeviz}, $\Omega^{\pm M}_{(m)}(\tau\otimes \tau)=0$ can be rewritten into
\begin{align*}
{\rm Res}_{z}z^{2m}e^{\xi(t'-t'',z)} \tau\left(\hat{t}'-\left[z^{-1}\right]_{\rm o}\right)\tau\left(\hat{t}''+\left[z^{-1}\right]_{\rm o}\right)=0,\quad m\geq0,
\end{align*}
where $\left[z^{-1}\right]_{\rm o}=\left(z^{-1},z^{-3}/3,\cdots\right)$, which is equivalent to
\begin{align}\label{KdVbilinear}
{\rm Res}_zz^{2m}e^{\hat{\xi}(\hat{t}'-\hat{t}'',z)} \tau\left(\hat{t}'-\left[z^{-1}\right]_{\rm o}\right)\tau\left(\hat{t}''+\left[z^{-1}\right]_{\rm o}\right)=0,\quad m\geq0.
\end{align}
Since the $t_{2n}$ part will contribute the even powers of $z$. This is exactly the bilinear equation of the KdV hierarchy \cite{JIMBO,Miwa2000}. There exists a differential operator
$$L=\partial^2+ u,\quad u=2\partial_x^2{\rm log}\tau(\hat{t}),\quad t_1=x,$$ such that
$$L_{t_n}=\left[\left(L^{n/2}\right)_{\geq0},L\right],$$
which is the Lax equation of KdV hierarchy \cite{Miwa2000}.

\section{Conclusions and Discussions}
There are two different constructions of $\widehat{sl}_2$--integrable hierarchies from the highest weight representations of Lie algebra $\widehat{sl}_2$, that is KW and DJKM constructions. Here in this paper, we have succeeded in showing the equivalence of these two different constructions for $\widehat{sl}_2$--integrable hierarchy in untwisted and twisted cases by using the language of lattice vertex algebras. Our main results are Theorem \ref{12} and Theorem \ref{theorem:KdV}. It is proved that the untwisted $\widehat{sl}_2$--integrable hierarchy is just the $1$--Toda lattice hierarchy, which is corresponding to the homogenous representation of $\widehat{sl}_2$. While KdV hierarchy is the twisted $\widehat{sl}_2$--integrable hierarchy constructed from the principal representation of $\widehat{sl}_2$. Note that it is usually quite difficult to prove the equivalence of bilinear equations directly, that is
$$\eqref{oKW}\Leftrightarrow\eqref{eflqxdsxxfc},\quad\eqref{KdVKW}\Leftrightarrow\eqref{KdVbilinear},$$
while the methods used in this paper provide one choice.

Just as we stated in the introduction part, the integrable hierarchies constructed by KW methods are usually presented in the forms of bilinear equations. It is usually very difficult to directly derive the Lax equations from the bilinear equations in KW construction. While in DJKM construction, there are many successful examples, such as KP and Toda hierarchies. Therefore the equivalent DJKM construction for the KW integrable hierarchies will be helpful in deriving Lax structures from bilinear equations. The results here can be generalized to $\widehat{sl}_n$--integrable hierarchy and provide one way to investigate relations of integrable hierarchies constructed by KW and DJKM methods.
\\
\\
\noindent{\bf Acknowledgements}:

We thank Professor Todor Milanov (IPMU, The University of Tokyo) for his long--term guide and help, and also thank Professor Bojko Bakalov (North Carolina State University) for his help in understanding twisted module over lattice vertex algebras. This work is supported by National Natural Science Foundation of China (Grant Nos. 12171472 and 12261072)
and ``Qinglan Project" of Jiangsu Universities.\\

\noindent{\bf Conflict of Interest}:

 The authors have no conflicts to disclose.\\

\noindent{\bf Data availability}:

Date sharing is not applicable to this article as no new data were created or analyzed in this study.

\end{document}